%% file: REPOSE.tex
\documentclass[conference]{IEEEtran}
\IEEEoverridecommandlockouts
\usepackage{cite}
\usepackage{amsmath,amssymb,amsfonts}
\usepackage{algorithmic}
\usepackage{graphicx}
\usepackage{textcomp}
\usepackage{xcolor}
\def\BibTeX{{\rm B\kern-.05em{\sc i\kern-.025em b}\kern-.08em
    T\kern-.1667em\lower.7ex\hbox{E}\kern-.125emX}}

\usepackage{multirow}
\usepackage{array}
\usepackage{tikz}
\usepackage{listings}
\usepackage{courier}
\usepackage{enumitem}
\usepackage{times}
\usepackage{booktabs}
\usepackage{subfigure}
\usepackage[ruled,vlined,linesnumbered]{algorithm2e}
\usepackage{amsthm}

\newtheorem{theorem}{Theorem}

\newtheorem{definition}{Definition}
\newtheorem{example}{Example}
\newtheorem{lemma}{Lemma}

\def\cmt{\textcolor{black}}
\begin{document}

\title{REPOSE: Distributed Top-$k$ Trajectory Similarity Search with Local Reference Point Tries}

\author{
	\IEEEauthorblockN{Bolong Zheng$^{1}$, Lianggui Weng$^1$, Xi Zhao$^1$, Kai Zeng$^2$, Xiaofang Zhou$^3$, Christian S. Jensen$^4$
	}
	\IEEEauthorblockA{$^1$Huazhong University of Science and Technology, Wuhan, China\\
		Email: \{bolongzheng, liangguiweng, zhaoxi\}@hust.edu.cn}
	\IEEEauthorblockA{$^2$Alibaba Group, Hangzhou, China\\
		Email: zengkai.zk@alibaba-inc.com}
	\IEEEauthorblockA{$^3$University of Queensland, Brisbane, Australia\\
		Email: zxf@itee.uq.edu.au}
	\IEEEauthorblockA{$^4$Aalborg University, Aalborg, Denmark\\
		Email: csj@cs.aau.dk}
}

\maketitle

\input{abstract}

\begin{IEEEkeywords}
trajectory similarity, top-$k$ query, distributed
\end{IEEEkeywords}

\input{introduction}
\input{problem}

\input{rptrie}
\input{queryprocessing}
\input{distributed}
\input{similarity}
\input{experiments}
\input{related}
\input{conclusion}

\bibliographystyle{abbrv}
\bibliography{myRef}

\input{appendix}

\end{document}

%% file: abstract.tex
\begin{abstract}
Trajectory similarity computation is a fundamental component in a variety of real-world applications, such as ridesharing, road planning, and transportation optimization. Recent advances in mobile devices have enabled an unprecedented increase in the amount of available trajectory data such that efficient query processing can no longer be supported by a single machine. As a result, means of performing distributed in-memory trajectory similarity search are called for. However, existing distributed proposals either suffer from computing resource waste or are unable to support the range of similarity measures that are being used.
We propose a distributed in-memory management framework called \textsf{REPOSE} for processing top-$k$ trajectory similarity queries on Spark. \cmt{We develop a reference point trie (RP-Trie) index to organize trajectory data for local search}. In addition, we design a novel heterogeneous global partitioning strategy to eliminate load imbalance in distributed settings.
We report on extensive experiments with real-world data that offer insight into the performance of the solution, and show that the solution is capable of outperforming the state-of-the-art proposals.
\end{abstract}

%% file: introduction.tex
\section{Introduction} \label{sec:introduction}
With the widespread diffusion of GPS devices (such as smart phones), massive amounts of data describing the motion histories of moving objects, known as trajectories, are being generated and managed in order to serve a wide range of applications, e.g., travel time prediction \cite{DBLP:conf/gis/PfoserBBUTT08,DBLP:conf/ssdbm/PfoserTV06}, taxis dispatching, and path planning \cite{DBLP:conf/huc/ZhangLZCSL11,DBLP:journals/ijkesdp/YamamotoUW10}. For example, Didi has released an open dataset that includes 750 million GPS points in Xi'an with a sampling rate of 2--4 seconds over the span of one month.

\cmt{Top-$k$ trajectory similarity search that finds $k$ trajectories that are most similar to a query trajectory is a basic operation in offline analytics applications}. In the context of massive trajectory data, it is non-trivial to enable efficient top-$k$ trajectory similarity search. 
Many existing studies\cite{DBLP:conf/time/DingTS08,DBLP:conf/icde/FrentzosGT07,DBLP:conf/gis/BakalovHT05,DBLP:conf/mdm/BakalovHKT05,DBLP:conf/icde/VlachosGK02,DBLP:conf/compgeom/DriemelS17,DBLP:conf/vldb/Keogh02,DBLP:conf/sigmod/FaloutsosRM94,DBLP:conf/kdd/VlachosHGK03,DBLP:journals/vldb/SuLZZZ20} focus on optimizing the query processing on a single machine. However, if the data cardinality exceeds the storage or processing capacity of a single machine, these methods do not extend directly to a distributed environment.

Instead, a distributed algorithm is called for that is able to exploit the resources of multiple machines. DFT \cite{DBLP:journals/pvldb/XieLP17} and DITA \cite{DBLP:conf/sigmod/Shang0B18} are state-of-the-art distributed trajectory similarity search frameworks.
They include global partitioning methods that place trajectories with similar properties in the same partition, and they use a global index to prune irrelevant partitions. Then, they merge the results of local searches on the surviving partitions. Finally, they return a top-$k$ result. However, these methods have two shortcomings that limit their use in practice.

(1) Computing resource waste. DITA and DFT aim to guarantee load balancing by means of their global partitioning strategies. In particular, DITA places trajectories with close first and last points in the same partition. DFT places trajectory segments with close centroids in the same partition.
However, only surviving partitions are employed on the distributed nodes, while compute nodes with no surviving partitions remain idle and are not utilized. This has an adverse effect on computing resource utilization.

(2) Limited support for similarity measures.
We observe that DFT supports Hausdorff, Frechet \cite{DBLP:journals/ijcga/AltG95}, and DTW \cite{DBLP:conf/icde/YiJF98}, but does not support LCSS \cite{DBLP:conf/icde/VlachosGK02}, EDR \cite{DBLP:conf/sigmod/ChenOO05}, and ERP \cite{DBLP:conf/vldb/ChenN04}. Further, DITA supports Frechet, DTW, EDR, and LCSS, but does not support Hausdorff and ERP. In order to support diverse application scenarios, it is important to be able to accommodate a wide range of similarity measures in a single system. 

We propose an efficient distributed in-memory management framework, \textsf{REPOSE}, for top-$k$ trajectory similarity query processing that supports a wide range of similarity measures.
To eliminate poor computing resource utilization, \textsf{REPOSE} includes a novel heterogeneous global partition strategy that places similar trajectories in different partitions. Therefore, partitions have similar composition structure with the effect that most partitions and compute nodes are likely to contribute to a query result, which improves the computing resource utilization significantly and enables a load balancing.

In addition, we propose a reference point trie (RP-Trie) index to organize the trajectories in each partition. During index construction, we convert each trajectory into a reference trajectory by adopting Z-order \cite{DBLP:conf/isaac/DaiS03} that preserves sketch information. Then, we build an RP-Trie based on this representation, which reduces the space consumption. During query processing, we traverse the local RP-Trie in a best-first manner. We first develop a one-side lower bound on internal nodes for pruning. Moreover, we devise a two-side lower bound on leaf nodes for further improvement. Finally, we design a pivot-based pruning strategy for metrics.

In summary, we make the following contributions:
\begin{itemize}
	\item We propose a distributed in-memory framework, \textsf{REPOSE}, for the processing of top-$k$ trajectory similarity queries on Spark. The framework supports multiple trajectory similarity measures, including Hausdorff, Frechet, DTW, LCSS, EDR, and ERP distances.
	\item We discretize trajectories into reference trajectories and develop a reference point trie to organize the reference trajectories. Several optimization techniques are proposed to accelerate query processing. 
    \item We design a novel heterogeneous global partitioning strategy to achieve load balancing and to accelerate query processing by balancing the composition of partitions.
	\item We compare the performance of \textsf{REPOSE} with the state-of-the-art distributed trajectory similarity search frameworks using real datasets. The experimental results indicate that \textsf{REPOSE} is able to outperform the existing frameworks.
\end{itemize}

The rest of the paper is organized as follows. 
Section \ref{sec:problem} states the problem addressed. 
Section \ref{sec:RPTrie} discusses how to discretize trajectories into reference trajectories and build RP-Trie.
Section \ref{sec:QueryProcessing} introduces the query processing and several optimization techniques. 
Section \ref{sec:distributed} introduces the global partition strategy.
Section \ref{sec:similarity} describes how to extend our algorithm to other similarity measures.
Section \ref{sec:experiment} presents the results of our experimental study. 
Section \ref{sec:related} reviews related work. 
Finally, Section \ref{sec:conclusion} concludes the paper.

%% file: problem.tex
\section{Problem Definition} \label{sec:problem}
\begin{table}
	\caption{Summary of Notations}
	\label{tb:notation}
	\centering
	\begin{tabular}{cl}
		\toprule
		\textbf{Notation} & \textbf{Definition}\\
		\midrule
		$\mathcal{D}$ & Trajectory dataset \\
		$\tau$ & A trajectory \\
		$\tau^*$ & A reference trajectory \\
		$\textsf{D}_\textsf{H}(\tau_1,\tau_2)$ & Hausdorff distance between $\tau_1$ and $\tau_2$\\
		$ \mathcal{A} $ & A square region that encloses all trajectories\\
		$ g $ & The grid \\
		$\delta$ & The grid side length \\
		$LB_o$ & One-side lower bound \\
		$LB_t$ & Two-side lower bound \\
		$LB_p$ & Pivot based lower bound \\
		$ N_p $ & The number of the pivot trajectories  \\ 
		\bottomrule
	\end{tabular}
\end{table}

We proceed to present the problem definition. Frequently used notation is summarized in Table \ref{tb:notation}.

\begin{definition}[Trajectory]
A trajectory $\tau$ is a finite, time-ordered sequence $\tau=\langle p_1,p_2,\dots,p_n \rangle$, where each $p_i \in \tau$ is a sample point with a longitude and a latitude.
\end{definition}

A range of distance functions have been used for quantifying the similarity between trajectories, including the Hausdorff, Frechet, DTW, LCSS, EDR and ERP distances. To ease the presentation, we initially focus on the Hausdorff distance, and extend the coverage to include other distances later.
\begin{definition}[Trajectory Distance]\label{def:hausdorff}
Given two trajectories $\tau_1=\langle q_1,q_2,\dots,q_m \rangle$ and $\tau_2=\langle p_1,p_2,\dots,p_n\rangle$, the Hausdorff distance between $\tau_1$ and $\tau_2$ is computed as follows.
\begin{equation}\label{dist_form1} 
\textsf{D}_\textsf{H}(\tau_1,\tau_2)=\max\{\max_{q_i \in \tau_1}\min_{p_j \in \tau_2}d(q_i,p_j), \max_{p_j \in \tau_2}\min_{q_i \in \tau_1}d(q_i,p_j)\},
\end{equation}
where $d(p_i,q_j)$ is the Euclidean distance.
\end{definition}

\begin{definition}[Trajectory Similarity Search]
Given a set of trajectories $\mathcal{D}=\{\tau_1,\tau_2,\dots,\tau_N\}$, a query trajectory $\tau_q$, a distance function \textsf{Dist}$(\cdot)$, and an integer $k$, the top-$k$ trajectory similarity search problem reports a set $R$ of $k$ trajectories, where $\forall \tau \in R$ and $\forall \tau' \in \mathcal{D}-R$, we have $\textsf{Dist}(\tau_q,\tau) < \textsf{Dist}(\tau_q,\tau')$.
\end{definition}

\begin{example}
\cmt{
Given a query trajectory $ \tau_q$ and the dataset $\mathcal{D}=\{\tau_1,\tau_2,\tau_3,\tau_4,\tau_5\}$ in Table \ref{tb:trajectory_table}, we process a top-$2$ query. Computing the Hausdorff distance between $ \tau_q$ and all trajectories in $\mathcal{D}$, we get $ \textsf{D}_\textsf{H}(\tau_q,\tau_1)=2.83$, $ \textsf{D}_\textsf{H}(\tau_q,\tau_2)=6.08$, $ \textsf{D}_\textsf{H}(\tau_q,\tau_3)=6.71$, $ \textsf{D}_\textsf{H}(\tau_q,\tau_4)=3.16$, $
\textsf{D}_\textsf{H}(\tau_q,\tau_5)=6.08  $. Therefore, the top-$ 2 $ result is $ \{\tau_1,\tau_4\} $.}
\end{example}

\begin{table}[t]
	\caption{Point Coordinates of Trajectories}\label{tb:trajectory_table}
	\centering \small
	\renewcommand\arraystretch{1.1}
	%\begin{tabular}{|m{2cm}|m{5.8cm}|m{5.8cm}|}
%	\begin{tabular}{|m{1.4cm}<{\centering}|m{6.5cm}<{\centering}|}
	\begin{tabular}{|c|l|}
		\hline 
		{\bf \centering Trajectory}  & \multicolumn{1}{c|}{\bf  \centering Point Coordinates} \\\hline
		$ {\tau}_1 $ & $ (0.5,7.5), (2.5,7.5), (6.5,7.5), (6.5,4.5) $ \\\hline
		$ {\tau}_2 $ & $ (1.5,0.5), (2.5,0.5), (2.5,4.5), (4.5,4.5) $ \\\hline
		$ {\tau}_3 $ & $  (4.5,0.5), (7.5,0.5), (7.5,2.5), (4.5,2.5), (4.5,1.5)  $  \\\hline
		$ {\tau}_4 $ & $  (0.5,7.5), (2.5,7.5), (5.5,7.5), (5.5,3.5) $  \\\hline
		$\tau_5$ & $ (1.5,0.5), (2.5,0.5), (2.5,5.5), (0.5,5.5), (0.5,2.5)$  \\\hline
		\cmt{$\tau_q$} & \cmt{$ (0.5,6.5), (2.5,6.5), (4.5,6.5) $}  \\\hline
	\end{tabular}
\end{table}

\begin{figure}[t]
	\centering
	\includegraphics[width=.3\textwidth]{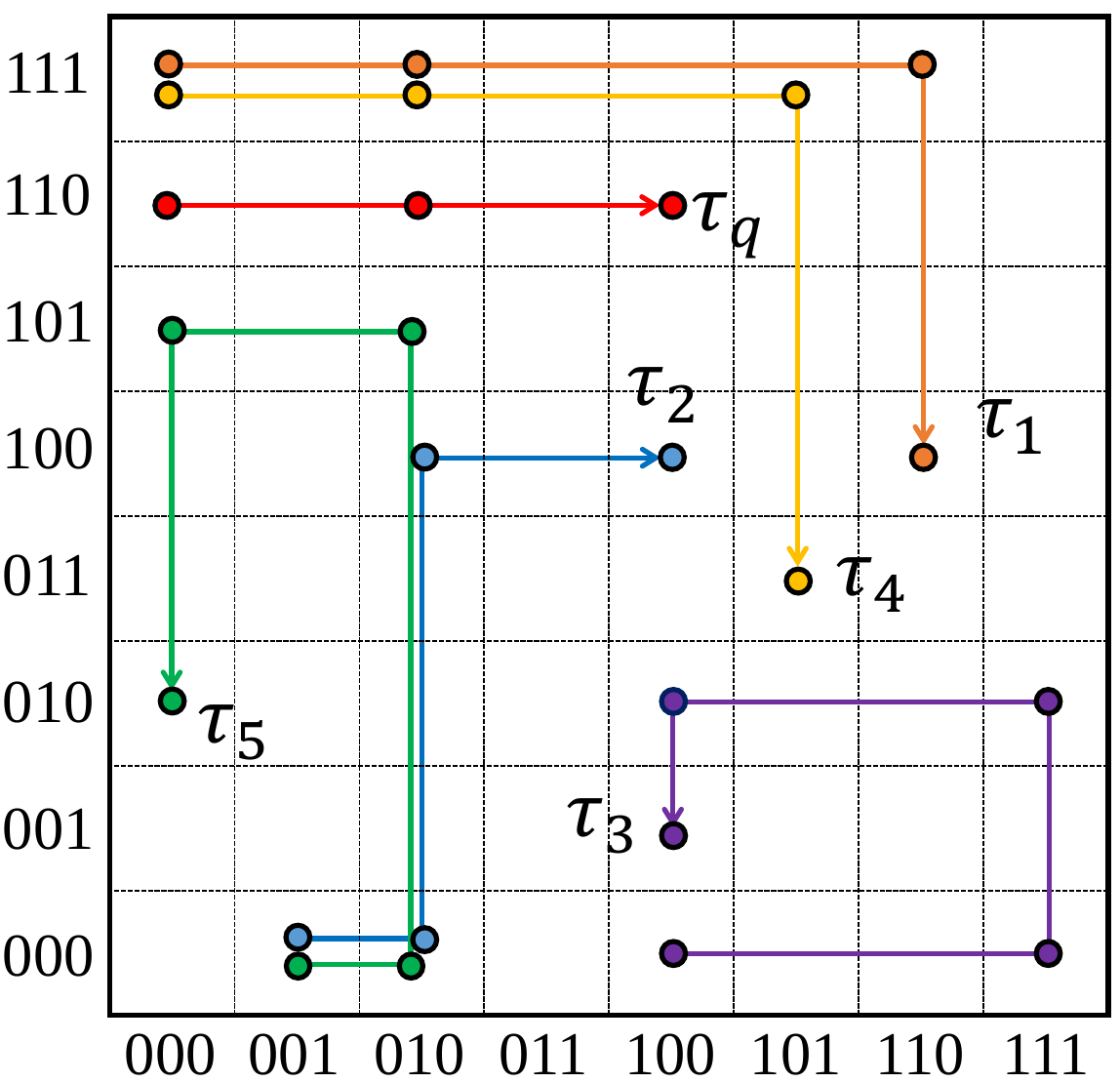} %1.png是图片文件的相对路径
	\caption{A Running Example} %caption是图片的标题
	\label{fig:grid} %此处的label相当于一个图片的专属标志，目的是方便上下文的引用
\end{figure}

%% file: rptrie.tex
\section{Reference Point Trie}\label{sec:RPTrie}
Before describing the distributed framework, we proceed to introduce a reference point trie (RP-Trie) index for local search. First, we explain how to convert trajectories into reference trajectories. Second, we show how to build an RP-Trie on the reference trajectories. Finally, we propose an order-independent optimization to improve the performance.

\subsection{Discretizing Trajectories with Z-order}
Inspired by signature-based trajectory representation \cite{DBLP:journals/tkde/TaLXLHF17}, we adopt the Z-order \cite{DBLP:conf/isaac/DaiS03} to map trajectories from their native two-dimensional space to a one-dimensional space. Let $\mathcal{A}$ be a square region with side length $U$ that encloses all trajectories. We partition $\mathcal{A}$ by means of a regular $l \times l$ grid with side length $\delta$, where $l=U/\delta$ is a power of 2. Each cell $g$ has a unique z-value and a reference point (the center point of $g$).

\begin{example}
\cmt{In Fig. \ref{fig:grid}, we show an 8 $ \times $ 8 grid}. The z-value of a cell is the alternate combination of its horizontal and vertical coordinates. For instance, the cell whose (binary) horizontal and vertical coordinates are \underline{010} and 101, respectively, has z-value  \underline{0}{1}\underline{1}{0}\underline{0}{1}. 
\end{example}

\begin{definition}[Reference Trajectory]\label{def:reference_trajectory}
\cmt{
We convert a trajectory $\tau=\langle  p_1,p_2,\dots,p_n\rangle$ into a reference trajectory $\tau^*=\langle p^*_1,p^*_2,\dots,p^*_n\rangle$, where $p^*_i$ is the reference point of the cell $g_i$ that $p_i$ belongs to.   
Note that $\tau^*$ corresponds to a sequence of of z-values $Z=\langle z_1,z_2,\dots,z_n\rangle$, where $z_i$ is the z-value of $g_i$.}
\end{definition}
Note that the grid granularity affects the fidelity of a reference trajectory. A small $ \delta $ ensures a high fidelity.
%In addition, different with other distance measures, only the Hausdorff distance is point order independent, so we can further simplify $\tau^*$ in two steps. (1) Z-value ordering: We sort the points in $\tau^*$ by their z-values. (2) Z-value deduplication. We save only one point when two or more points in $\tau^*$ have identical z-value.

\subsection{Building an RP-Trie Index}
We proceed to cover how to build an RP-Trie on a set of reference trajectories, which is similar to building a classical trie index. We thus use the z-value as the value of the trie node and insert all reference trajectories into the trie. The difference is that if one reference trajectory is a prefix of another reference trajectory, we append the character $\$$ at the end, which guarantees that every reference trajectory ends at a leaf node. 

\begin{figure}[t]
	\centering
	\includegraphics[width=.4\textwidth]{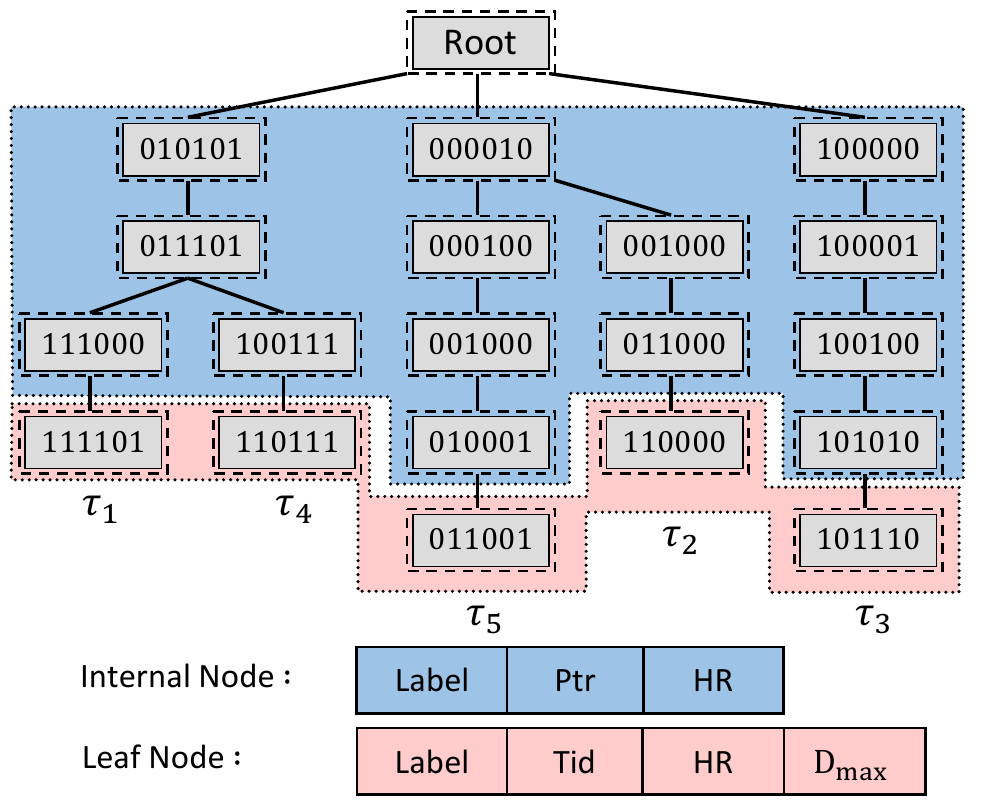} %1.png是图片文件的相对路径	
	\caption{The Structure of RP-Trie} %caption是图片的标题
	\label{fig:trie} %此处的label相当于一个图片的专属标志，目的是方便上下文的引用
\end{figure}

The structure of an RP-Trie is shown in Fig. \ref{fig:trie}.
In internal nodes, the $ \textsf{Label}$ attribute is the z-value of the node that corresponds to the coordinates of the reference point, and $ \textsf{Ptr} $ is a pointer array pointing to child nodes. In leaf nodes, the $ \textsf{Label} $ attribute has the same meaning as internal nodes.
Since each reference trajectory ends at a leaf node and represents multiple trajectories, we use $ \textsf{Tid} $ to record the trajectory ids. 

\textbf{Pivot trajectory.}
For similarity metrics such as Hausdorff, Frechet, and ERP, the triangle inequality can be used for pruning.
To enable this, we use pivot trajectories to estimate the lower bound distance between a query trajectory and all trajectories in a subtree rooted at a node. Specifically, we select $ N_p $ trajectories as global pivot trajectories. For each node, let $\mathcal{T}_{sub}$ be the set of reference trajectories covered by it. We then keep an $N_p$-dimensional array $ \textsf{HR} $, where $ \textsf{HR}[i] $ stores a tuple $(min, max)$ that represents the minimum and maximum distances between all reference trajectories in $\mathcal{T}_{sub}$ and the $ i $-th pivot trajectory.

In particular, the selection of pivot trajectories has a substantial effect on the pruning performance. Thus, existing proposals consider how to find optimal pivots \cite{DBLP:journals/cacm/Shapiro77,DBLP:journals/prl/BustosNC03}. In general, we aim to obtain a pivot set such that any two pivot trajectories in it are as distant as possible. With this in mind, we adopt a practical but effective method \cite{DBLP:conf/adbis/SkopalPS04}. We uniformly and randomly sample $m$ groups of $N_p$ trajectories. In each group, we compute the distances of any two trajectories, and let the sum of all distances be the score of the group. Finally, we choose the $N_p$ trajectories in the group with the largest score as the set of pivot trajectories.

In addition, we store a value $ \textsf{D}_{max} $ for each leaf node that is the maximum distance between the node's reference trajectory $\tau^*$ and the trajectories in the leaf node. 

\cmt{
\textbf{Cost analysis.}
Assume that we have $N$ reference trajectories with maximum length $L$. As the RP-Trie has at most $N \cdot L$ nodes, the space cost is $O(N \cdot L \cdot N_p)$. To construct an RP-Trie, the cost of the computation of the distances between pivots and all trajectories dominates all other costs and takes $O(N \cdot L^2 \cdot N_p)$. Since both $L$ and $N_p$ are small ($N_p=5$ is used in experiments), both the space and time costs are affordable.}

\textbf{Succinct trie structure.}
We observe that the upper levels of an RP-Trie consist of few nodes that are accessed frequently, while the lower levels comprise the majority of nodes that are accessed relatively infrequently due to pruning. Inspired by SuRF \cite{DBLP:journals/sigmod/ZhangLLAKP19}, we introduce a fast search structure for the upper levels of an RP-Trie and a space-efficient structure for the lower levels. Specifically, we optimize the RP-Trie by switching between bitmaps and byte arrays at different layers. 
For each upper level node, we use two bitmaps $ B_c$ and $B_l $ with the same size as the number of grid cells, to separately record the value of a child node and the state of the child node. If a cell is the node's child node, we set the corresponding bit in $ B_c $ to 1. If a child node is not a leaf node, we set the corresponding bit in $ B_l $ to 1. Then $ B_c$ and $B_l $ of all nodes are concatenated in breadth-first order separately such that we can quickly access any upper level node.
For each lower level node, we serialize the structure with byte sequences. Given that the trie becomes sparse in these levels, it is space-efficient to use byte sequences.

\subsection{Optimization by Z-value Re-arrangement}
\cmt{Given $\tau_q$ and a reference trajectory $\tau^*$, if we interchange the positions of any two points in $\tau^*$ to generate a new reference trajectory $\tau^{*'}$, we have $\textsf{D}_\textsf{H}(\tau_q,\tau^*)=\textsf{D}_\textsf{H}(\tau_q,\tau^{*'})$.
Therefore, we know that, unlike other distance measures, Hausdorff is order independent. Hence, we propose an optimization that reduces the size of an RP-Trie by rearranging the reference trajectories that achieve longer common prefixes. Specifically, we simplify $\tau^*$ in two steps: (1) z-value deduplication: We keep only one point when two or more points in $\tau^*$ have identical z-value. (2) z-value re-ordering: We re-order the points in $\tau^*$.
For example, assume that an RP-Trie is constructed from trajectories $ \tau_2 $ and $ \tau_5 $ as shown to the left in Fig. \ref{fig:trieOp}. If the z-values 000100 and 001000 of $ \tau_5 $ are swapped, the total number of nodes decreases, as shown to the right in Fig. \ref{fig:trieOp}.}

\begin{figure}[t]
	\centering
	\includegraphics[width=.4\textwidth]{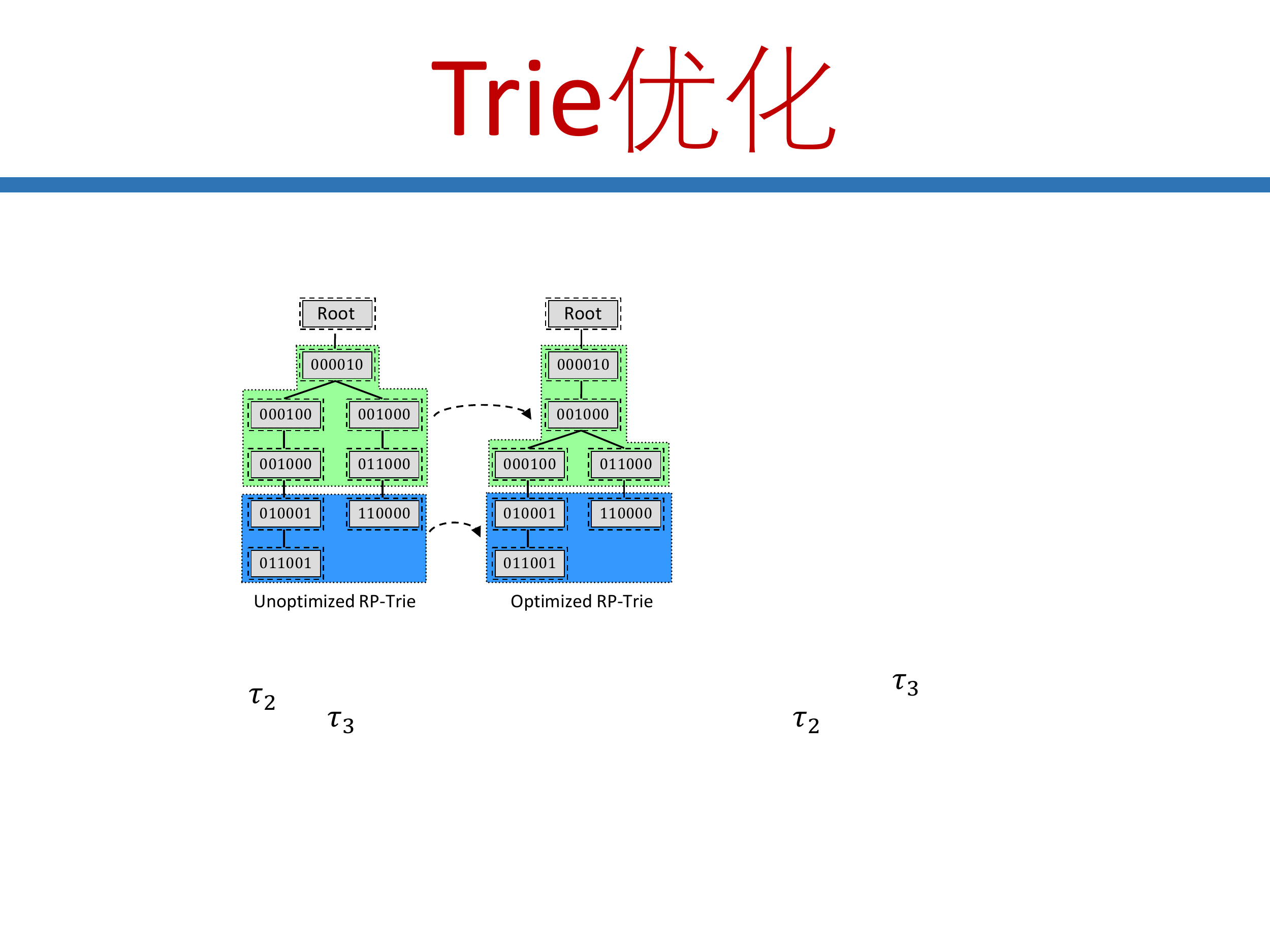} %1.png是图片文件的相对路径
	\caption{Optimized Reference Point Trie} %caption是图片的标题
	\label{fig:trieOp} %此处的label相当于一个图片的专属标志，目的是方便上下文的引用
\end{figure}

Thus, we aim to build an optimized RP-Trie with fewer nodes by means of z-value re-arrangement.
Let $Z$ be the set of z-values of a reference trajectory. Let $\mathcal{Z}=\{Z_1,Z_2,\dots,Z_N\}$ be the collection of z-values sets of all trajectories. 
In order to find the optimized RP-Trie structure, 
we want each level of the trie to have the minimum number of nodes. In other words, we need to find the smallest cell set $ G_{opt}=\{g_1,g_2,\dots,g_r\} $
such that for any $Z_i$ in $\mathcal{Z}$, $Z_i$ contains at least one cell in $G_{opt}$. 
This way, we regard $ g_1,g_2,\dots,g_r $ as child nodes of the root and partition $ \mathcal {Z} $ into classes $\mathcal{C}_1,\mathcal{C}_2,\dots,\mathcal{C}_r$. 
Using this method to recursively partition each class $ \mathcal{C}_i $, we obtain an optimized RP-Trie.

However, it is difficult to find the set $ G_{opt}$ for $\mathcal{Z}$.
Before explaining how to solve this problem, we introduce a well known NP-hard problem: the hitting set problem.
\begin{definition}[$HS({\mathcal{Z}},b)$]
	Given a collection of sets $\mathcal{Z}=\{Z_1,$ $Z_2,…,Z_N\}$ and a budget $b$, the hitting set $HS({\mathcal{Z}},b)$ returns a set $ G $ that satisfies for any $Z_i$ in $\mathcal{Z}$,
	$$
	|G|\leq b \land G\cap Z_i\neq \emptyset
	$$
	If no such set $G$ exists, $HS({\mathcal{Z}},b)$ returns an empty set.
\end{definition}
Set $G$ is called a hitting set of $\mathcal{Z}$. Usually, a collection of sets $\mathcal{Z}$ has more than one hitting set. It is easy to find a hitting set for it when disregarding the budget $b$. However, it is costly to obtain a hitting set within a budget $b$.

\begin{theorem}
Finding an optimized RP-Trie is NP-hard.
\end{theorem}
\begin{proof}
To find an optimized RP-Trie for our reference trajectory collection $\mathcal{Z}$, we find the smallest cell set $ G_{opt}$ with size $|G_{opt}|=r$ such that for any $Z_i$ in $\mathcal{Z}$, $Z_i$ contains at least a cell in $G_{opt}$. Obviously, $G_{opt}$ is the hitting set of $\mathcal{Z}$ with the minimum size. So, the result of $HS({\mathcal{Z}},r)$ is $G_{opt}$, and the result of $HS({\mathcal{Z}},r-1)$ is an empty set. Therefore, finding $G_{opt}$ is reduced to solving $HS({\mathcal{Z}},r)$ and $HS({\mathcal{Z}},r-1)$. This completes the proof.
\end{proof}

We employ a greedy algorithm to solve this problem. We make the most frequent z-value $ z_1 $ in $ \mathcal {Z} $ the first child node of the root. Then we put all $ Z_i $ containing $ z_1 $ into the subtree of $ z_1 $ and remove them from $ \mathcal {Z} $. Next, we make the current most frequent z-value $ z_2 $ in $ \mathcal {Z} $ the second child node and place all $ Z_i $ containing $ z_2 $ in its subtree and remove them from $ \mathcal {Z} $.  Repeating the above process until $ \mathcal {Z} $ is empty, we obtain the division structure of the current level of the trie. We recursively repeat the process for the remaining levels to build an optimized RP-Trie.
\cmt{Assume that we have $N$ reference trajectories and $M$ cells. The computation cost of the greedy algorithm is $O(N \cdot M^2)$, where $M$ is usually a small constant. The details can be found in the accompanying technical report \cite{repose-report}.}

%\todo{Assume that we have $N$ reference trajectories and $M$ grids. Let $L$ be the maximum length of reference trajectories in $\mathcal{Z}$ and $L_\mathcal{Z}=\sum_{i=1}^{N} |Z_i|$ be the total number of z-values in $\mathcal{Z}$. The optimized RP-Trie has at most $L+1$ layers and $L_\mathcal{Z}+1$ nodes. For each node, we need to count the frequency of each z-value in its subtree and store them by an array with size $M$. The accumulated counting time is $O(L_\mathcal{Z}L+L_\mathcal{Z}M)$ and the extra space consumption to store the frequency arrays is $O(L_\mathcal{Z}M)$. Usually, $L, M$ are not large constants and $L\leq M$. Obviously, $ L_\mathcal{Z}\leq NL$ and thus the building time and extra space consumption of the optimized trie are both $O(NM^2)$. When $M$ is small, we can regard them as $O(N)$.}

%% file: queryprocessing.tex
\section{Query Processing}\label{sec:QueryProcessing}
We proceed to present the top-$k$ query processing as well as several query optimization techniques. 
\cmt{For ease of presentation, we mainly discuss the details of the Hausdorff distance. Extensions to support other distance measures, such as Frechet and DTW, can be found in Section \ref{sec:similarity}.}
First, we introduce our basic algorithm to perform a top-$ k $ query, where the nodes in the RP-Trie are traversed in a best-first manner. We develop a pruning condition called one-side lower bound to filter out unqualified internal nodes. We provide three optimization techniques. The first technique reduces the computing cost on a node by using intermediate computation results. The second technique provides a tight pruning condition called two-side lower bound for pruning on leaf nodes. The third technique utilizes the metric property for pivot based pruning.

\subsection{Search Procedure}
To answer a top-$k$ query, the nodes in the RP-Trie are traversed in ascending order of a lower bound, which is called one-side lower bound $ LB_o $ used to prune unqualified internal nodes. Let $d_k$ be the $ k $-th smallest found result. An internal node is pruned when its $LB_o$ is greater than $d_k$.

We briefly introduce the procedure as follows: 
\begin{enumerate}
	\item  We construct a priority queue $ E $ to keep the nodes based on the ascending order of their lower bounds. We initially insert the root node into $E$. The $minHeap$ stores the trajectories with distance to $\tau_q$ no larger than $d_k$.
	\item\label{item:search_procedure_2} 
	The head element $t$ is popped from $E$. If the lower bound of $t$ is smaller than $d_k$, we go to Step \ref{item:search_procedure_3}). Otherwise, we stop the procedure and return $minHeap$.
	\item\label{item:search_procedure_3} If node $t$ is a leaf node, we compute the distances between $\tau_q$ and the trajectories recorded by $ \textsf{Tid} $. Then, we update $d_k$ and $minHeap$; Otherwise, we compute the $LB_o $ for each child node of $ t $ and insert it into $E$.
	\item Repeating Steps \ref{item:search_procedure_2}) and \ref{item:search_procedure_3}) when $ E $ is not empty. Otherwise, we stop the procedure and return $minHeap$.
\end{enumerate}

We introduce how to compute the one-side lower bound $LB_o$. For simplicity, for a sub reference trajectory that starts from the root node and ends at an internal node, we also call it a reference trajectory. Next, we give a formal definition of the one-side lower bound.

\begin{definition}[One-Side Lower Bound]
Given the query trajectory $ \tau_q = \langle q_1, q_2,\dots,q_m \rangle $ and a node whose reference trajectory is $\tau^*=\langle p^*_1,p^*_2,\dots,p^*_n \rangle $, one-side lower bound of the node is defined as follows,
\begin{equation}
LB_o(\tau_q, \tau^*)=\max\{ \max _{p^*_j \in \tau^*} \min _{q_i \in \tau_q} d(q_i,p^*_j)-\frac{\sqrt{2}\delta}{2},0\} 
\end{equation}	
\end{definition}
Two following lemmas ensure that the returned results by above procedure are the correct top-$ k $ query results of $\tau_q$.
\begin{lemma}
The $ LB_o $ of a leaf node is smaller than the minimum distance between $\tau_q$ and any trajectory stored in the node.
\label{tLB_o1}
\end{lemma}
\begin{proof}
Given $ \tau_q = \langle q_1, q_2,\dots,q_m\rangle $, a node's reference trajectory $\tau^*= \langle p^*_1,p^*_2,\dots,p^*_n\rangle$ and any trajectory $\tau=\langle p_1,p_2,\dots,p_{n'} \rangle$ stored in the node, we prove that $ \textsf{D}_\textsf{H}(\tau_q,\tau) \geq LB_o(\tau_q,\tau^*) $. 
Assume $p_i$ falls into the cell represented by $ p^*_k $ (we call it $p_i\in p^*_k$), we have $d(q_j,p_i)\geq \max\{d(q_j,p^*_k)-\sqrt{2}\delta/2,0\}$ according to the triangle inequality.
Then we have
\begin{equation*}                                                         
\begin{split}
\textsf{D}_\textsf{H}(\tau_q,\tau)
=&\max \{ \max _{q_j \in \tau_q} \min _{p_i \in \tau} d(p_i,q_j),\max _{p_i \in \tau} \min _{q_j \in \tau_q} d(p_i,q_j) \}\\
\geq& \max _{p_i \in \tau} \min _{q_j \in \tau_q} d(p_i,q_j)\\
\geq& \max _{p_i \in \tau,p_i\in p^*_k} \min _{q_j \in \tau_q} \max\{d(q_j,p^*_k)-\frac{\sqrt{2}\delta}{2},0\}\\
=&\max _{p^*_k \in \tau^*} \min _{q_j \in \tau_q} \max\{d(q_j,p^*_k)-\frac{\sqrt{2}\delta}{2},0\}\\
=&\max\{ \max _{p^*_k \in \tau^*} \min _{q_j \in \tau_q} d(q_j,p^*_k)-\frac{\sqrt{2}\delta}{2},0\} \\
=&LB_o(\tau_q,\tau^*)
\end{split}              
\end{equation*}	

%\begin{equation*}                                                         
%\begin{split}
%\textsf{D}_\textsf{H}(\tau_q,\tau)
%=&\max \{ \max _{q_j \in \tau_q} \min _{p_i \in \tau} d(p_i,q_j),\max _{p_i \in \tau} \min _{q_j \in \tau_q} d(p_i,q_j) \}\\
%\geq& \max _{p_i \in \tau} \min _{q_j \in \tau_q} d(p_i,q_j)\\
%\geq& \max _{p_i \in \tau,p_i\in p^*_k} \min _{q_j \in \tau_q} \max\{d(q_j,p^*_k)-\sqrt{2}\delta/2,0\}\\
%=&\max _{p^*_k \in \tau^*} \min _{q_j \in \tau_q} \max\{d(q_j,p^*_k)-\sqrt{2}\delta/2,0\}\\
%=&\max\{ \max _{p^*_k \in \tau^*} \min _{q_j \in \tau_q} d(q_j,p^*_k)-\sqrt{2}\delta/2,0\} \\
%=&LB_o(\tau_q,\tau^*)
%\end{split}              
%\end{equation*}
\end{proof}

Lemma \ref{tLB_o1} indicates that if the node is a leaf node, $ LB_o $ is the lower bound distance between $ \tau_q $ and the trajectories recorded by $ \textsf{Tid} $. Thus, the trajectories recorded by the node can be safely pruned when $ LB_o \geq d_k $. However, if the node is not a leaf node, Lemma \ref{tLB_o1} does not work, so it comes to the following Lemma \ref{tLB_o2}. Lemma \ref{tLB_o2} shows that $LB_o $ of a node is larger than that of its parent node. If $ LB_o $ is larger than $ d_k $, no top-$ k $ results exist in the subtree of current node, and we can safely prune it. 

\begin{lemma}
For an internal node, the $LB_o$ between its child node and $\tau_q$ is greater than that between the node and $\tau_q$.
\label{tLB_o2}
\end{lemma}

\begin{proof}
Let $ \tau_p^* = \langle p^*_1,\cdots, p^*_{n-1} \rangle $ and $ \tau^* = \langle p^*_1,\cdots, p^*_n \rangle $ be the reference trajectories of a node and its child node, respectively. 
Given $ \tau_q $, we prove that $LB_o(\tau_q,\tau^*)\geq LB_o(\tau_q,\tau_p^*)$.
%Given $ \tau_q = \langle q_1, q_2,\dots,q_m \rangle$, we prove that $LB_o(\tau_q,\tau^*)\geq LB_o(\tau_q,\tau_p^*)$.
	\begin{equation*}
	\begin{split}
	LB_o(\tau_q,\tau^*)
	=&\max _{p^*_k \in \tau^*} \min _{q_j \in \tau_q} \max\{d(q_j,p^*_k)-\frac{ \sqrt{2}\delta}{2},0\}\\
	=&\max\{\max _{p^*_k \in \tau_p^*} \min _{q_j \in \tau_q} \max\{d(q_j,p^*_k)-\frac{ \sqrt{2}\delta}{2},0\},\\ 
	& \min _{q_j \in \tau_q} \max\{d(q_j,p^*_n)-\frac{ \sqrt{2}\delta}{2},0\}\}\\
	\geq& \max _{p^*_k \in \tau_p^*} \min _{q_j \in \tau_q} \max\{d(q_j,p^*_k)-\frac{ \sqrt{2}\delta}{2},0\}\\
	=&LB_o(\tau_q,\tau_p^*)
	\end{split}
	\end{equation*}
\end{proof}

\cmt{Note that Lemmas \ref{tLB_o1} and \ref{tLB_o2} also hold for Frechet and DTW with a modification regarding the computation of $LB_o$. The details can be found in Section \ref{sec:similarity}.}

\subsection{Pruning using Two-Side Lower Bound}
We know that one-side lower bound is designed for pruning the subtree of internal nodes. For each internal node, we only have the prefix reference trajectories from the root node to this internal node. In contrast, for a leaf node, we have complete reference trajectories of the trajectories stored in this node. Therefore, a tight lower bound, called two-side lower bound, is provided to further improve the query efficiency.
Based on the reference trajectory, the two-side lower bound is computed as follows:
\begin{definition}[Two-Side Lower Bound]	
	Given $ \tau_q$ and a leaf node whose reference trajectory is $\tau^*$, the two-side lower bound $ LB_t $ is computed as follows,
	\begin{equation}
	LB_t(\tau_q, \tau^*)=\max \left\{\textsf{D}_\textsf{H}(\tau_q,\tau^*)-\textsf{D}_{max},0  \right\}
	\end{equation}
where $\textsf{D}_{max} $ is the maximum distance from all trajectories in the current node to the reference trajectory.
\end{definition}

According to the triangle inequality, it is easy to prove that $ LB_t $ is the lower bound distance of all the trajectories stored in current node. In other word, $ LB_t $ is smaller than the distance between $\tau_q$ and any trajectories stored in the leaf node. 
The focuses of $ LB_o $ and $ LB_t $ are different. $LB_o$ is able to prune all nodes of a subtree, but its value is loose compared with $LB_t$. While $ LB_t $ is designed for pruning trajectories in a leaf node, its value is close to exact distance.
Next, we introduce how to reduce the computation overheads of both $ LB_o $ and $LB_t$ by using the intermediate computation results.
%Similar to $ LB_o $, we can reduce the computation overhead of $LB_t$ to $O(m)$ by using the intermediate computation results, such as $r_1$, $r_2$, and $c_{max}$.

\begin{algorithm}[t]
	\caption{\textsf{CompLB}($\tau_q,p^*,r,c_{max}$)}  
	\label{CalLowerBound} 
	\LinesNumbered
	\KwIn{The query $ \tau_q $, the newly added reference point $ p^* $, an array $r$ with size $|\tau_q|$ and $c_{max}$}
	\KwOut{$ LB_o $,
		$ LB_t $,
		the updated $r$ and $c_{max}$}
	$ i \gets 1 $; $ r_{max} \gets 0 $; $ c~\gets +\infty $\;
	\While{$i \leq |\tau_q|$}{
		$ dist=d( \tau_q[i],p^*) $\;
		$ r[i]=\min\{dist,r[i]\} $\;
		$ c=\min\{dist,c\} $\;
		$ r_{max}=\max\{r[i],r_{max}\} $\;
		$ i \gets i+1$
	}
	$ c_{max} \gets \max\{c_{max},c\} $\;
	$ LB_o=\max\{c_{max}-\frac{\sqrt{2}\delta}{2},0\} $\;
	$ LB_t=\max\{\max \left\{r_{max},c_{max} \right\}-\textsf{D}_{max},0\}$\;
	\Return $ LB_o, LB_t, r, c_{max} $\;
\end{algorithm}

\subsection{Optimization using Intermediate Results}
To access each node on the RP-Trie, the computational overheads of $LB_o$ and $LB_t$ are both $ O(mn) $, where $m$ and $n $ are the lengths of the query trajectory and the reference trajectory, respectively. We propose an optimization that uses the intermediate computation results of the parent node to reduce the overhead to $ O (m) $.

Before we describe the optimization, we first review the steps of computing Hausdorff. We assume that the lengths of $ \tau_q $ and the reference trajectory $ \tau^* $ are both 2. As shown in Fig. \ref{fig:hausdorff}, $ q_1$ and $q_2 $ are points in the query trajectory, while $p^*_1$ and $p^*_2 $ are points in the reference trajectory. Let ${r_i} $ be the minimum value in row $i$, e.g., the minimum of $d(q_i,p^*_1)$ and $d(q_i,p^*_2)$. Let ${c_j} $ be the minimum value in column $j$, e.g., the minimum of $d(q_1,p^*_j)$ and $d(q_2,p^*_j)$. Let $ c_{max} $ be the maximum value among all $ {c_j} $. According to the definition of Hausdorff,
$$LB_o(\tau_q, \tau^*)=\max\{c_{max}-\frac{\sqrt{2}\delta}{2},0\},$$
and
$$
LB_t(\tau_q, \tau^*)=  \max\{\max \left\{ \max\{r_i\},c_{max} \right\}-\textsf{D}_{max},0\}.
$$
After accessing $p^*_2$, we record ${r_1,r_2}$ and $c_{max}$ as intermediate results. For a new reference trajectory $\tau_{new}^*=\langle p^*_1,p^*_2,p^*_3\rangle$, both $LB_t(\tau_q, \tau^*_{new})$ and $LB_o(\tau_q, \tau^*_{new})$ can be quickly computed by $d(q_1,p^*_3)$, $d(q_2,p^*_3)$ and intermediate results, since we only have to update the values of $c_{max}$ and all $r_i$, where
\begin{equation}
\begin{split}
r'_1=&\min\{r_1,d(q_1,p^*_3)\} \\
r'_2=&\min\{r_2,d(q_2,p^*_3)\}\\
c'_{max}=&\max\{c_{max},\min\{d(q_1,p^*_3),d(q_2,p^*_3)\}\}.
\end{split}
\end{equation}
Similarly, ${r'_1,r'_2}$ and $c'_{max}$ are the intermediate results after accessing $p^*_3$ that can be used for the subsequent computation.

Therefore, the computational overheads of $LB_o$ and $LB_t$ for each node are $ O(m) $. 
%In addition, $LB_o$ and $LB_t$ can be obtained through one computation. 
Algorithm \ref{CalLowerBound} describes the process that we compute $ LB_o $ and $LB_t$ in detail. 
%We observe that the main computational overhead of $ LB_o $ comes from the time of computing ${c_3} $. For a query trajectory with size $m$, the computational overhead of $ LB_o $ is reduced to $ O(m) $.
%First, we update the array $ r_{min} $ which records the minimum value of each row in the distance matrix. Then, we update $c_{max}$ and obtain $LB_o$.  

\begin{figure}
	\centering
	\includegraphics[width=1.0\linewidth]{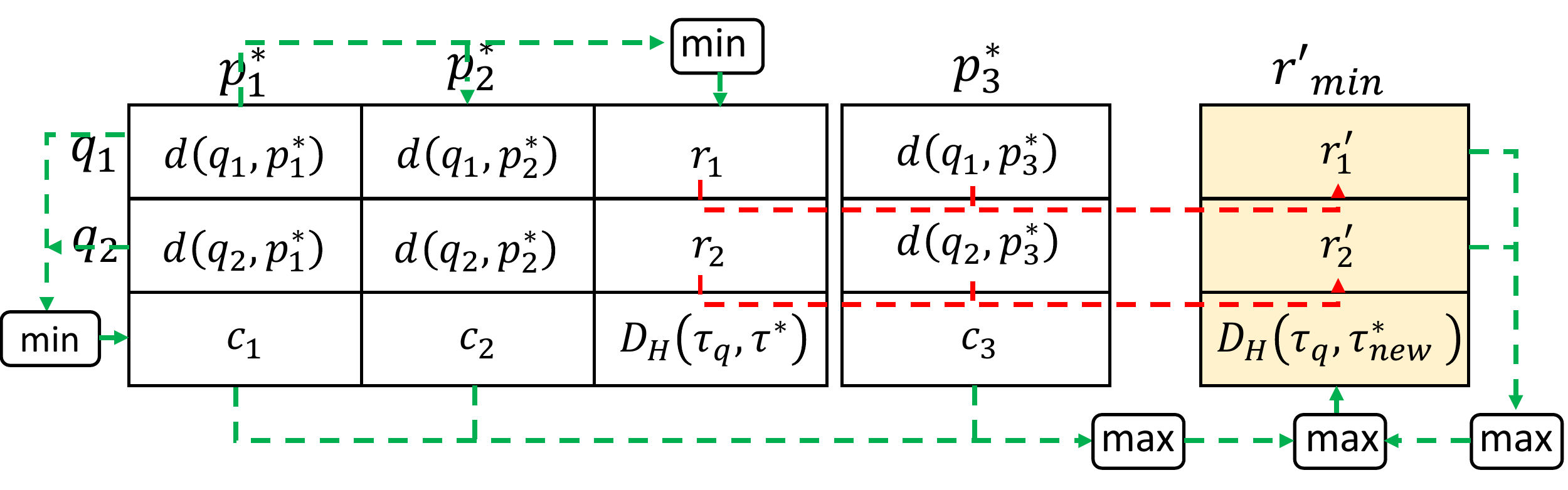}
	\caption{Distance Matrix}
	\label{fig:hausdorff}
\end{figure}

\subsection{Pivot based Pruning}\label{sec:QueryProcessing_pivot}
As the Hausdorff distance is a metric, the pivot-based pruning strategy\cite{DBLP:journals/pvldb/ChenGZJYY17} is used to further improve the query efficiency. Therefore, we propose a pruning strategy by using a pivot based lower bound $ LB_p $.

During preprocessing, we select $ N_p $ trajectories as pivot trajectories. For each node in RP-Trie, we store a distance array $\textsf{HR}$. During the query processing, we first compute the distances $d_{qp} $ between $\tau_q $ and all pivot trajectories. The time complexity is $O(N_p \cdot m \cdot n) $, where $m$ and $n $ are the lengths of $ \tau_q $ and pivot trajectory respectively. 

We compute the pivot based lower bound $ LB_p $ as follows,
\begin{equation}\label{eq:lb_p}
	LB_p=\max\{|d_{qp}[i]-\textsf{HR}[i].max-\frac{\sqrt{2}}{2}\delta|\}.
\end{equation}
Since $N_p $ is small, the time cost of computing $LB_p $ is small.

%% file: distributed.tex
\section{Distributed Framework}\label{sec:distributed}

We proceed to cover the distributed framework. First, we discuss the drawbacks of existing global partitioning strategies. Then, we introduce our partitioning strategy that aims to ensure better computing resource utilization during query processing. Finally, we discuss how to build a local RP-Trie based on Spark RDDs.

\subsection{Drawbacks of Existing Global Partitioning Strategies}
A straightforward partitioning method is to randomly divide trajectories into different partitions such that similar numbers of trajectories are kept in each partition. However, this does not guarantee load balancing because the query times in different partitions can be quite different. 
Existing proposals, such as DITA, DFT, and DISAX \cite{DBLP:journals/tkde/YagoubiAMP20}, employ a homogeneous partitioning method that groups similar trajectories into the same partition. Then, they use a global index to prune irrelevant partitions without trajectories that are similar to the query trajectory. However, these methods have limited performance for two main reasons: 
\cmt{
\begin{itemize}
\item \textbf{Computing resource waste}. After global pruning, only compute nodes with surviving partitions compute local results, while other nodes are idle and do not contribute. 
%The result is a waste of computing resources.
\item \textbf{Less effective local pruning}. If the trajectories in each partition are similar, the pruning in local search is less effective since the MBRs of the local index may have large overlaps.
\end{itemize}
The homogeneous partitioning strategy targets batch search (e.g., high concurrency), where different partitions response to different query trajectories. However, not all partitions of trajectories are involved in the batch query processing if the queries are skewed. For example, ride-hailing companies tend to issue a batch of analysis queries in hot regions to increase profits. Therefore, some computing resource waste cannot be avoided.}

\subsection{Our Global Partitioning Strategy}\label{sec:distributed_method}
An effective global partitioning strategy should ensure load balancing among partitions and should accelerate local-search part of query processing. Unlike existing strategies, we propose a heterogeneous partitioning strategy that places similar trajectories in different partitions based on the following observations:
\begin{itemize}
\item As mentioned, the pruning in local search is less effective in existing methods. 
Inspired by the principle of maximum entropy, a high degree of data discrimination in a partition is likely to enable better pruning.
\item The composition structure of each partition is similar, since each partition contains trajectories from different groups of similar trajectories. It is possible that most partitions contain query results. Therefore, the query times of most partitions are similar, which enables load-balancing.
\end{itemize}

\cmt{Specifically, we use the simple but effective clustering algorithm SOM-TC \cite{DBLP:conf/icdcs/DewanGS17} to partition trajectories. 
We encode each trajectory $\tau$ as a reference trajectory $\tau^*$ using geohash. If $\tau_1^*=\tau_2^*$, we group $\tau_1$ and $\tau_2$ into a cluster. Note that we vary the space granularity of geohash to control the number of clusters. At first, we set the space granularity to a small value such that a cluster contains only one trajectory. Then we enlarge the space granularity gradually and group trajectories into larger clusters. This process stops when we reach about $N/N_{G}$ clusters, where $N$ is the dataset cardinality and $ N_{G} $ is the number of partitions.
We sort the trajectories based on their cluster id and trajectory id. Finally, we assign the trajectories from the sorted sequence to partitions in round-robin fashion.}

\subsection{Building Local Index based on RDD}
Spark Core \cite{DBLP:conf/nsdi/ZahariaCDDMMFSS12} is based on the concept of abstract Resilient Distributed Dataset (RDD), which supports two types of operations: transformation and action operations. 
However, RDDs are designed for sequential processing, and random accesses through RDDs are expensive since such accesses may simply turn into full scans over the data.
In our method, random access is inevitable. We solve this challenge by exploiting an existing method \cite{DBLP:journals/pvldb/XieLP17}, that is, we package the data and index into a new class structure and use the class as the type of RDD.

Specially, the abstract class $ \sf{Partitioner} $ is provided by Spark for data partitioning, and Spark allows users to define
their own partitioning strategies through inheritance from the abstract class. We use $ \sf{Partitioner} $ to implement the strategy presented in Section \ref{sec:distributed_method}. Then, we package all trajectories within an RDD partition and RP-Trie index into a class structure called  $ \sf{RpTraj} $ that is defined as follow.
\lstset{
 	basicstyle = \footnotesize\ttfamily
}
\begin{lstlisting}[columns=fullflexible ]
case class RpTraj (trajectory: Array, Index: RP-Trie)
\end{lstlisting}
We change the type of RDD to $ \sf{RpTraj} $.
\begin{lstlisting}[columns=fullflexible ]
type RpTrieRDD = RDD [RpTraj] 
\end{lstlisting}
The $\sf{RpTrieRDD}$ is an RDD structure. The transformation operation $ \sf{mapPartitions} $ is provided by Spark, and we use $ \sf{RpTrieRDD.mapPartitions} $ to manipulate each partition.
In particular, in each partition, we first use $ \sf{Index.build} $ to construct RP-Trie index. Then we use $ \sf{Index.query} $ to execute the top-$ k $ query. Finally, the master collects the results from each partition by $ \sf{collect} $ and determines the global top-$ k $ result.

%% file: similarity.tex
\section{Other Similarity measures}\label{sec:similarity}
We proceed to extend our algorithm to Frechet and DTW. Due to the different properties of similarity measures, the index structures are distinguished with that for Hausdorff.

\subsection{Extension on Frechet}

The only difference with Hausdorff is that Frechet is sensitive to the order of points. So the RP-Trie is used without the trie optimization. As Frechet is also a metric, we still use pivot trajectories to accelerate the query processing.
Given $\tau_q=\{q_1,q_2,\cdots,q_m\}$ and $\tau=\{p_1,p_2,\cdots,p_n\}$, the Frechet distance between $\tau_q$ and $\tau$ is computed as follows,
\begin{equation} \label{e_frechet} 
	\begin{split}
		&\textsf{D}_\textsf{F}(\tau_q, \tau)=\\
		&\begin{cases}
			\!\max _{j=1}^{n} {d}\left(q_{1}, p_{j}\right)   &  m=1 \\
			\!\max _{i=1}^{m} {d}\left(q_{i}, p_{1}\right)  & n=1 \\
			\!\max \!\left\{
			{d}\left(q_{m}, p_{n}\right), \min \left\{
			\textsf{D}_\textsf{F}\left(\tau_q^{m-1},\tau^{n-1}\right), \right.\right. & \\
			\!\left.\left. \textsf{D}_\textsf{F}\left(\tau_q^{m-1}, \tau\right), \textsf{D}_\textsf{F} \left(\tau_q, \tau^{n-1}\right)\right\}
			\right\}  & \text { otherwise } 
		\end{cases}
	\end{split}
\end{equation}
where $ \tau^{m-1} $ is the prefix of $ \tau $ by removing the last point.

We modify the computation of $ LB_o^{\textsf{F}}$ and $LB_t^{\textsf{F}}$ of Frechet, 
\begin{equation}
	LB_o^{\textsf{F}}(\tau_q, \tau^*)=\max\{ c_{min}-\frac{ \sqrt{2}\delta}{2},0\} 
\end{equation}
\begin{equation}\label{eq:lbt_frechet}
	LB_t^{\textsf{F}}(\tau_q, \tau^*)=\max\{ {f_{m,n}}-\frac{\sqrt{2}\delta}{2},0\} 
\end{equation}
where $f_{i,j}=\textsf{D}_\textsf{F}\left(\tau_q^i,\tau^{*j}\right)$ is the element of the $i$-th row and the $j$-th column of the distance matrix, and $ c_{min} $ is the minimum value of the newly added column, e.g., $ c_{min}=\min\left\lbrace f_{1,4},f_{2,4},f_{3,4}\right\rbrace $, as shown in Fig. \ref{fig:frechet}.
Fortunately, we are still able to quickly compute $ LB_o^{\textsf{F}}$ and $LB_t^{\textsf{F}}$ through the intermediate results. Given the current last column values $ \left\lbrace f_{1,3},f_{2,3},f_{3,3} \right\rbrace $, the values in the new column are computed as follows,
\begin{equation}
f_{i,4}=\max \left\{
{d}\left(q_{i}, p^*_{4}\right), \min \left\{f_{i-1,3},f_{i-1,4},f_{i,3} \right\} \right\}.
\end{equation}
%Since $\textsf{D}_\textsf{F}(\tau_q, \tau^*)=f_{34} $, we can obtain $ LB_o^{\textsf{F}}$ and $LB_t^{\textsf{F}}$ in the mean time.

\begin{figure}
	\centering
	\includegraphics[width=.65\linewidth]{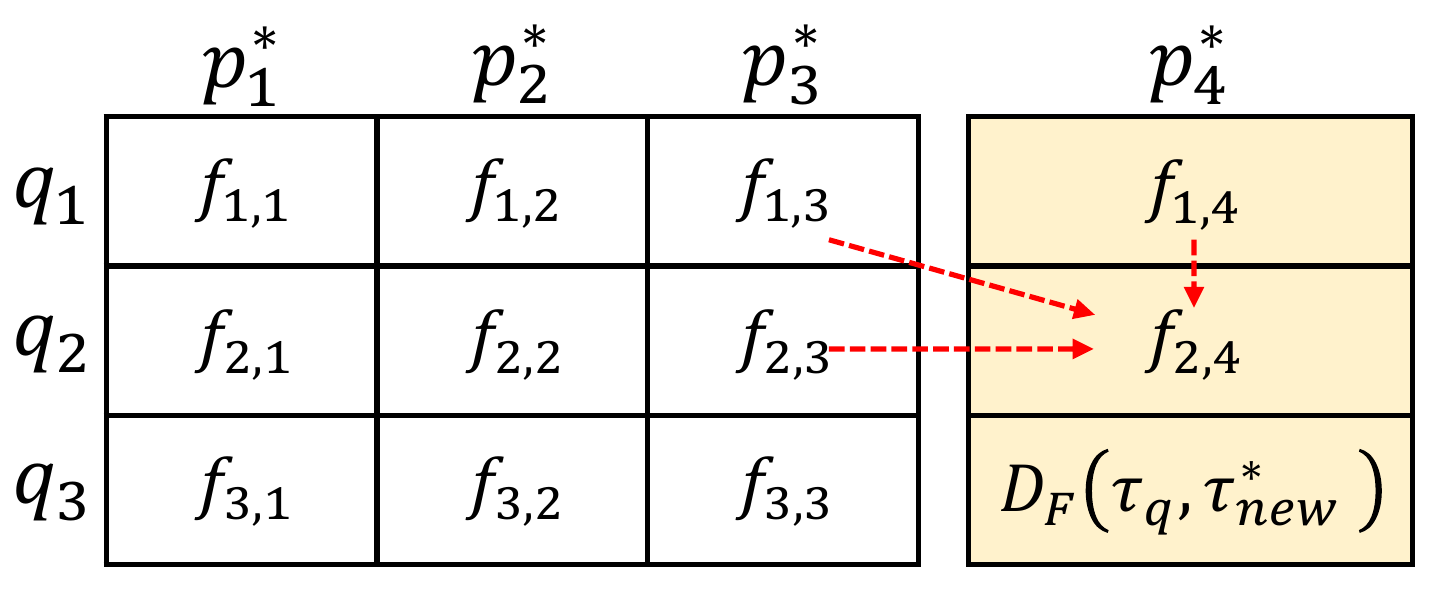}	
	\caption{Distance Matrix of Frechet}
	\label{fig:frechet}
\end{figure}

In addition, as Frechet is a metric, we can still use the pivot-based lower bound $ LB_p^{\textsf{F}}$ for pruning.
\begin{lemma}\label{le:frechet}
	$LB_o^{\textsf{F}}$, $LB_t^{\textsf{F}}$ and $LB_p^{\textsf{F}}$ have the same properties as those for Hausdorff.
	\begin{enumerate}
	\item	The $LB_o^{\textsf{F}}$, $LB_t^{\textsf{F}}$ and $LB_p^{\textsf{F}}$ of a leaf node are all smaller than the minimum Frechet distance between $\tau_q$ and any trajectory stored in the node.
	\item	For an internal node, the $LB_o^{\textsf{F}}$ between its child node and $\tau_q$ is larger than that between the node and $\tau_q$.
	\end{enumerate} 
\end{lemma}
The first property indicates that $LB_o^{\textsf{F}}$, $LB_t^{\textsf{F}}$ and $LB_p^{\textsf{F}}$ are correct lower bounds. The second property allows us to prune the sub-tree of an internal node with $LB_o^{\textsf{F}}$. %We prove these two properties.
\begin{proof}
For property 1, since $LB_t^{\textsf{F}}(\tau_q, \tau^*)\geq LB_o^{\textsf{F}}(\tau_q, \tau^*)$, we only need to prove that $LB_t^{\textsf{F}}$ and $LB_p^{\textsf{F}}$ are correct lower bounds. According to Eq. \ref{eq:lbt_frechet} and ${f_{m,n}}=\textsf{D}_\textsf{F}(\tau_q, \tau^*)$, we have
$$
LB_t^{\textsf{F}}(\tau_q, \tau^*)=\max\{ \textsf{D}_\textsf{F}(\tau_q, \tau^*)-\frac{\sqrt{2}\delta}{2},0\}.
$$
For a trajectory $\tau$ with a reference trajectory $\tau^*$, we have $\textsf{D}_\textsf{F}(\tau, \tau^*)\leq\frac{\sqrt{2}\delta}{2}$. Therefore, due to the triangle inequality, we have $LB_t^{\textsf{F}}(\tau_q, \tau^*)\leq \textsf{D}_\textsf{F}(\tau_q, \tau)$. Similarly, it is easy to have that $LB_p^{\textsf{F}}(\tau_q, \tau^*)\leq \textsf{D}_\textsf{F}(\tau_q, \tau)$.

For property 2, we prove that $LB_o^{\textsf{F}}(\tau_q, \tau^{*j}) \geq LB_o^{\textsf{F}}(\tau_q, \tau^{*{j-1}})$, $2 \leq j \leq n$. According to Eq. \ref{eq:lbt_frechet},
\begin{equation}
		\begin{split}
			LB_o^{\textsf{F}}(\tau_q, \tau^{*j})&=\max\{ \min{f_{i,j}}-\frac{ \sqrt{2}\delta}{2},0\} \\
			LB_o^{\textsf{F}}(\tau_q, \tau^{*{j-1}})&=\max\{ \min{f_{i,j-1}}-\frac{ \sqrt{2}\delta}{2},0\} 
		\end{split}
\end{equation}
Assume that $h = \arg\min_i f_{i,j-1}$, we have 
\begin{equation}
		\begin{split}
			f_{1,j}&=\max \left\{
			{d}\left(q_{1}, p^*_{j}\right), f_{1,j-1} \right\} \\
			&\geq f_{1,j-1}\geq f_{h,j-1}\\
			f_{2,j}&=\max \left\{
			{d}\left(q_{1}, p^*_{j}\right), \min \left\{f_{1,j-1},f_{1,j},f_{2,j-1} \right\} \right\} \\
			&\geq \min \left\{f_{1,j-1},f_{1,j},f_{2,j-1}\right\}
			\geq f_{h,j-1}
		\end{split}
\end{equation}
Therefore, we have ${f_{i,j}}\geq f_{h,j-1}, 1 \leq i \leq m$, and $\min{f_{i,j}}\geq f_{h,j-1}$. Finally, $LB_o^{\textsf{F}}(\tau_q, \tau^{*j})\geq LB_o^{\textsf{F}}(\tau_q, \tau^{*{j-1}})$.
\end{proof}

\subsection{Extension on DTW}
Since DTW is not a metric and is sensitive to the order of trajectory points, thus only the basic RP-Trie structure is used.
The DTW between $\tau_q$ and $\tau$ is computed as follows, 

\begin{equation} \label{e_dtw}
	\begin{split}
		&{\textsf{D}_\textsf{DTW}}(\tau_q, \tau)=\\
		&\begin{cases}
			\sum_{j=1}^{n} {d}\left(q_{1}, p_{j}\right) &  m=1 \\
			\sum_{i=1}^{m} {d}\left(q_{i}, p_{1}\right) & n=1 \\
			{d}\left(q_{m}, p_{n}\right)+\min \left\{
			{\textsf{D}_\textsf{DTW}}\left(\tau_q^{m-1},\tau^{n-1}\right), \right. & \\ \left.
			{\textsf{D}_\textsf{DTW}}\left(\tau_q^{m-1}, \tau\right),  {\textsf{D}_\textsf{DTW}} 
			\left(\tau_q, \tau^{n-1}\right)
			\right\} & \text { otherwise }
		\end{cases}
	\end{split}
\end{equation}
%For DTW, the reference trajectory $\tau^*$ of $\tau$ is defined as that in Def. \ref{def:reference_trajectory}.
The computations of $ LB_o^{\textsf{DTW}} $ and $ LB_t^{\textsf{DTW}} $ are similar to those of Frechet, which are computed as follows. 
\begin{equation}
	LB_o^{\textsf{DTW}} (\tau_q, \tau^*)=c_{min}
\end{equation}
\begin{equation}
	LB_t^{\textsf{DTW}} (\tau_q, \tau^*)={f_{i,n}}
\end{equation}
where $f_{i,j}=\textsf{D}_\textsf{DTW}\left(\tau_q^i,\tau^{*j}\right)$ is the element of the $i$-th row and the $j$-th column of the distance matrix, and $ c_{min} $ is the minimum value of the newly added column, i.e., $c_{min}=\min_{i} \{f_{i,j}\}$ at the $j$-th column.
We update $f_{i,j}$ as follows.
\begin{equation}
	f_{i,j}=
	{d'}(q_{i}, p^*_{j})+\min \left\{f_{i-1,j-1},f_{i-1,j},f_{i,j-1} \right\}.
\end{equation}
\cmt{Note that ${d'}(q_{i}, p^*_{j})$ is the minimum distance between $q_i$ and the cell that $p_j^*$ belongs to. We use it instead of $d(q_i,p_j^*)$ since the triangle inequality does not apply to DTW.}

\begin{lemma}
	$LB_o^{\textsf{DTW}}$, $LB_t^{\textsf{DTW}}$ have the same properties as those for Hausdorff.
	\begin{enumerate}
		\item	The $LB_o^{\textsf{DTW}}$, $LB_t^{\textsf{DTW}}$ of a leaf node are smaller than the minimum DTW distance between $\tau_q$ and any trajectory stored in the node.
		\item	For an internal node, the $LB_o^{\textsf{DTW}}$ between its child node and $\tau_q$ is larger than that between the node and $\tau_q$.
	\end{enumerate} 
\end{lemma}

The proof is similar to that of Lemma \ref{le:frechet}.

\cmt{The basic RP-Trie and the heterogeneous partitioning strategy are applicable to other distance measures. For metrics, such as ERP, the RP-Trie with the trie optimization and the pivot-based pruning can be employed in a similar way as for Frechet. For non-metrics, such as LCSS and EDR, the basic RP-Trie can be used similarly as for DTW.}

%% file: experiments.tex
\section{EXPERIMENTS}\label{sec:experiment}
%We report on extensive experiments with real datasets that offer insight into the performance of \textsf{REPOSE}.

\subsection{Experimental Setup}
All experiments are conducted on a cluster with 1 master node and 16 worker nodes. Each node has a 4-core Intel Xeon (Cascade Lake) Platinum 8269CY @ 2.50GHz processor and 32GB main memory and runs Ubuntu 18.04 LTS with Hadoop 2.6.0 and Spark 2.2.0.

\textbf{Datasets}. We conduct experiments on 3 types of datasets.
\begin{enumerate}
	\item Small scale and small spatial span: San Francisco (SF)\footnote{http://sigspatial2017.sigspatial.org/giscup2017/home}, Porto\footnote{https://www.kaggle.com/c/pkdd-15-predict-taxiservice-trajectory-i}, Rome\footnote{http://crawdad.org/roma/taxi/20140717}, T-drive \cite{DBLP:journals/tkde/YuanZXS13}.
	\item Large scale and small spatial span: Chengdu and Xi'an\footnote{https://gaia.didichuxing.com}.
	\item Large scale and large spatial span: OSM\footnote{https://www.openstreetmap.org}.
\end{enumerate}

The dataset statistics are shown in Table \ref{tb:dataset}. In the preprocessing stage, we remove the trajectories with length smaller than 10, and we split the trajectories with length larger than 1,000 into multiple trajectories. We uniformly and randomly select 100 trajectories as the query set.

\begin{table}[htb] 
	\caption{Statistics of Datasets}\label{tb:dataset}
	\centering \small
	\renewcommand\arraystretch{1.1}
	\begin{tabular}{|m{1.1cm}<{\centering}|m{1.5cm}<{\centering}|m{1.0cm}<{\centering}|m{1.8cm}<{\centering}|m{1.4cm}<{\centering}|}
		%	\begin{tabular}{|c|c|c|c|c|c|}
		\hline 
		{\bf \centering Datasets} & {\bf \centering Cardinality} & {\bf \centering AvgLen} & {\bf \centering Spatial span}  & {\bf \centering Size (GB)} \\ \hline
		T-drive & 356,228 & 22.6 &(1.89$^{\circ}$, 1.17$^{\circ}$)&0.16\\\hline
		SF & 343,696& 27.5&(0.54$^{\circ}$, 0.76$^{\circ} $)&0.19\\\hline
		Rome & 99,473& 152.4 &(1.21$^{\circ}$, 0.86$^{\circ}$)&0.28\\\hline
		Porto & 1,613,284& 48.9&(11.7$^{\circ}$, 14.2$^{\circ}$) &1.24\\\hline
		Xi'an & 6,645,727 & 230.1&(0.09$^{\circ}$, 0.08$^{\circ}$) &26.8\\\hline
		Chengdu & 11,327,466 & 188.9 &(0.09$^{\circ}$, 0.07$^{\circ}$)&37.7\\\hline
		OSM & 4,464,399 & 596.3&(360$^{\circ}$, 180$^{\circ}$)&50.8\\\hline	
	\end{tabular}
\end{table}

\textbf{Competing algorithms}. We compare the performance of \textsf{REPOSE} with three baseline algorithms:
\begin{enumerate}
	\item \textbf{DFT}: A segment-based distributed top-$k$ query algorithm that supports Hausdorff, Frechet, and DTW. Specifically, we choose the variant DFT-RB+DI that has the best query performance and the largest space overhead.
	
	\item \textbf{DITA}: An in-memory distributed range query algorithm. In order to support top-$k$ query, DITA first estimates a threshold and then executes a range query to find candidate trajectories. Finally, DITA finds the $k$ most similar trajectories. DITA does not support Hausdorff.
	
	\item \textbf{LS}: Brute-force linear search (LS) that computes the distances between the query and all trajectories in each partition and merges the results.
\end{enumerate}

\textbf{Parameter settings}. 
For DFT, we set the partition pruning parameter $ C = 5 $. For DITA, we set $ N_L = 32$ and the pivot size is set to $4$. The pivot selection strategy is the neighbor distance strategy. For \textsf{REPOSE}, due to the different spatial spans and data densities of the datasets, we choose different values of $\delta$.
For the SF, Porto, and Roma, we set $\delta=0.05$. For T-drive, we set $\delta=0.15$. For OSM, we set $\delta=1.0$. Next, for the Chengdu dataset, we set $\delta=0.01$ for Hausdorff, and $\delta=0.02$ for Frechet and DTW. Finally, for the Xi'an dataset, we set $\delta=0.01$ for Hausdorff, and $\delta = 0.03$ for Frechet and DTW. We choose $N_p=5$ pivot trajectories.

For queries, we use $k=100$ as the default.
\cmt{Since we have 16 worker nodes, each of which has a 4-core processor, we set the default number of partitions to 64 (each core processes a partition).}

\textbf{Performance metrics}. We study three performance metrics: (1) query time (QT), (2) index size (IS), and (3) index construction time (IT). We repeat each query 20 times and report the average values.
\cmt{
The index construction time of \textsf{REPOSE} includes the time for converting trajectories to reference trajectories, clustering the trajectories, and building the trie.
}

\subsection{Performance Evaluation}
We compare the algorithms on all datasets and three similarity measures (Hausdorff, Frechet, DTW). A performance overview is shown in Table \ref{tb:performance_overview}.
Due to the space limitation, we only report the performances on the T-drive, Xi'an, and OSM datasets and for the Hausdorff and Frechet distances in the remaining experiments.

\renewcommand\arraystretch{1.3}
\begin{table}[t]
	\centering
	\caption{Performance Overview}
	%\vspace*{0.1in}
	\label{tb:performance_overview}
	\scriptsize
	%\footnotesize
	\resizebox{0.5\textwidth}{!}{%
		\begin{tabular}{|m{0.65cm}<{\centering}|m{0.85cm}<{\centering}|m{0.9cm}<{\centering}|m{0.4cm}<{\centering}|m{0.55cm}<{\centering}|m{0.55cm}<{\centering}|m{0.8cm}<{\centering}|m{0.65cm}<{\centering}|m{0.8cm}<{\centering}|m{0.6cm}<{\centering}|}
			\hline
			\textbf{Metric}&\textbf{Distance}&\textbf{Algorithm}& \textbf{SF} & \textbf{Porto} & \textbf{Rome} & \textbf{T-drive} &\textbf{Xi'an}&\textbf{Chengdu}&\textbf{OSM} \\
			\hline
			\hline
			\multirow{12}*{\textbf{QT (s)}}&\multirow{4}*{Hausdorff}&\textbf{\textsf{REPOSE}}&\textbf{1.51} &\textbf{1.84} &\textbf{4.51} &\textbf{1.29} &\textbf{19.14} &\textbf{30.55} &\textbf{32.16} \\			
			\cline{3-10}
			& &\textbf{DITA}& / &/ &/ &/ &/ &/ &/ \\
			\cline{3-10}	
			& &\textbf{DFT}&4.03 &32.91 &14.99 &2.50 &1857.88 &2462.54 &273.68 \\
			\cline{3-10}
			& &\textbf{LS}&2.46 &1.97 &4.64 &1.95 &88.68 &119.11 &393.36 \\
			\cline{3-10}
			\cline{2-3}	
			& \multirow{4}*{Frechet}&\textbf{\textsf{REPOSE}}&\textbf{1.73} &\textbf{1.88} &\textbf{6.87} &\textbf{1.28} &\textbf{43.27} &\textbf{55.39} &\textbf{35.61} \\
			\cline{3-10}	
			& &\textbf{DITA}&2.67 &1.92 &15.30 &2.39 &91.31 &62.87 &94.12 \\
			\cline{3-10}	
			& &\textbf{DFT}&3.89 &32.03 &14.99 &2.33 &2100.80 &2284.49 &418.40\\
			\cline{3-10}
			& &\textbf{LS}&2.88 &2.41 &7.32 &2.20 &119.89 &156.85 &634.78\\
			\cline{3-10}
			\cline{2-3}	
			& \multirow{4}*{DTW}&\textbf{\textsf{REPOSE}}&\textbf{1.98} &\textbf{2.23} &\textbf{6.99} &\textbf{1.48} &\textbf{180.09} &\textbf{189.52} &\textbf{499.52}\\	
			\cline{3-10}
			& &\textbf{DITA}&3.15 &2.32 &11.06 &2.78 &186.36 &199.98 &509.53\\	
			\cline{3-10}
			& &\textbf{DFT}&3.88 &32.05 &14.96 &2.36&2472.35 &2300.35 &525.38\\
			\cline{3-10}
			& &\textbf{LS}&2.37 &2.35 &7.04 &1.92 &202.46 &225.92 &608.89\\
			\cline{3-10}
			\cline{2-3}	
			\cline{1-2}
			\multirow{12}*{\textbf{IS (GB)}}&\multirow{4}*{Hausdorff}&\textbf{\textsf{REPOSE}}&\textbf{0.19} &\textbf{1.68} &\textbf{0.28} &\textbf{0.16} &\textbf{28.01} &\textbf{39.40} &\textbf{47.92} \\			
			\cline{3-10}
			& &\textbf{DITA}&/ &/ &/ &/ &/ &/ &/ \\
			\cline{3-10}	
			& &\textbf{DFT}&1.03 &6.10 &1.39 &0.85 &142.30 &213.71 &235.06 \\
			\cline{3-10}
			& &\textbf{LS}&/ &/ &/ &/ &/ &/ &/ \\
			\cline{3-10}
			\cline{2-3}	
			& \multirow{4}*{Frechet}&\textbf{\textsf{REPOSE}}&\textbf{0.22} &\textbf{1.71} &\textbf{0.31} &\textbf{0.17} &32.13 &56.26 &60.39\\
			\cline{3-10}	
			& &\textbf{DITA}&0.23 &1.81 &0.32 &0.21 &\textbf{31.43} &\textbf{44.52} &\textbf{58.60}\\
			\cline{3-10}	
			& &\textbf{DFT}&1.03 &6.10 &1.39 &0.85 &142.30 &213.71 &235.06\\
			\cline{3-10}
			& &\textbf{LS}&/ &/ &/ &/ &/ &/ &/ \\
			\cline{3-10}
			\cline{2-3}	
			& \multirow{4}*{DTW}&\textbf{\textsf{REPOSE}}&\textbf{0.22} &\textbf{1.71} &\textbf{0.31} &\textbf{0.17} &32.13 &56.26 &60.39\\	
			\cline{3-10}
			& &\textbf{DITA}&0.23 &1.81 &0.32 &0.21 &\textbf{31.43} &\textbf{44.52} &\textbf{58.60}\\	
			\cline{3-10}
			& &\textbf{DFT}&1.03 &6.10 &1.39 &0.85 &142.30 &213.71 &235.06\\
			\cline{3-10}
			& &\textbf{LS}&/ &/ &/ &/ &/ &/ &/\\
			\cline{3-10}
			\cline{2-3}	
			\cline{1-2}
			\multirow{12}*{\textbf{IT (s)}}&\multirow{4}*{Hausdorff}&\textbf{\textsf{REPOSE}}&\textbf{3.96} &\textbf{5.16} &\textbf{4.40} &\textbf{3.47} &\textbf{75.41} &\textbf{119.08} &\textbf{100.96}\\			
			\cline{3-10}
			& &\textbf{DITA}&/ &/ &/ &/ &/ &/ &/\\
			\cline{3-10}	
			& &\textbf{DFT}&17.28 &75.80 &20.11 &15.32 &1407.29 &2071.53 &3389.59\\
			\cline{3-10}
			& &\textbf{LS}&/ &/ &/ &/ &/ &/ &/ \\
			\cline{3-10}
			\cline{2-3}	
			& \multirow{4}*{Frechet}&\textbf{\textsf{REPOSE}}&\textbf{5.27} &\textbf{5.19} &\textbf{5.34 }&\textbf{4.30} &\textbf{97.61} &\textbf{148.63} &\textbf{148.98} \\
			\cline{3-10}	
			& &\textbf{DITA}&20.24 &20.42 &18.21 &21.54 &235.02 &323.38 &349.31 \\
			\cline{3-10}	
			& &\textbf{DFT}&17.28 &75.80 &20.11 &15.32 &1407.29 &2071.53 &3389.59\\
			\cline{3-10}
			& &\textbf{LS}&/ &/ &/ &/ &/ &/ &/ \\
			\cline{3-10}
			\cline{2-3}	
			& \multirow{4}*{DTW}&\textbf{\textsf{REPOSE}}&\textbf{5.32} &\textbf{5.01} &\textbf{5.52} &\textbf{4.10} &\textbf{99.71} &\textbf{152.33} &\textbf{149.08}\\	
			\cline{3-10}
			& &\textbf{DITA}&17.99 &21.15 &17.94 &18.28 &240.75 &326.72 &363.95\\	
			\cline{3-10}
			& &\textbf{DFT}&19.98 &75.81 &23.71 &14.12 &1417.10 &2061.31 &3379.60\\
			\cline{3-10}
			& &\textbf{LS}&/ &/ &/ &/ &/ &/ &/ \\
			\cline{1-10}
			\hline	
		\end{tabular}
	}
\end{table}

\textbf{Performance overview}. 
We make the following observations from Table \ref{tb:performance_overview}: 
(1) In terms of query time and index construction time, \textsf{REPOSE} significantly outperforms the baseline methods. 
\cmt{
\textsf{REPOSE} also maintains a smaller index than DFT and DITA in most cases. \textsf{REPOSE} has a slightly larger index than DITA on Chengdu, Xi'an, and OSM for the Frechet and DTW distances.
This occurs because DITA compresses all trajectories into fixed-length representative trajectories. Nevertheless, \textsf{REPOSE} can capture trajectory features adaptively and is able to improve the query efficiency.
}
(2) The index size of DFT is about 4 times those of \textsf{REPOSE} and DITA on all datasets. For example on Xi'an, DFT requires 142GB, while \textsf{REPOSE} and DITA only require 32GB. 
The reason is that DFT needs to regroup line segments into trajectories when computing distances and needs a dual index that takes up extra space.
(3) For the first two types of datasets, it is interesting to see that the query time of LS is smaller than those of DITA and DFT. 
\cmt{
The reason is that the datasets are small, so the pruning capabilities are not utilized fully. In addition, some index-specific operations add computational overhead. For large datasets with small spatial span, the data density is high, which renders pruning more challenging.
}

\begin{figure}[t]
	\centering
	\includegraphics[width=.4\textwidth]{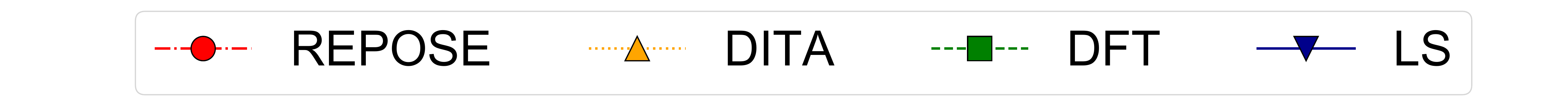} %1.pdf是图片文件的相对路径	
	\subfigure[T-drive with Hausdorff]{
		\begin{minipage}[c]{0.45\linewidth}
			\centering
			\includegraphics[width=1\textwidth]{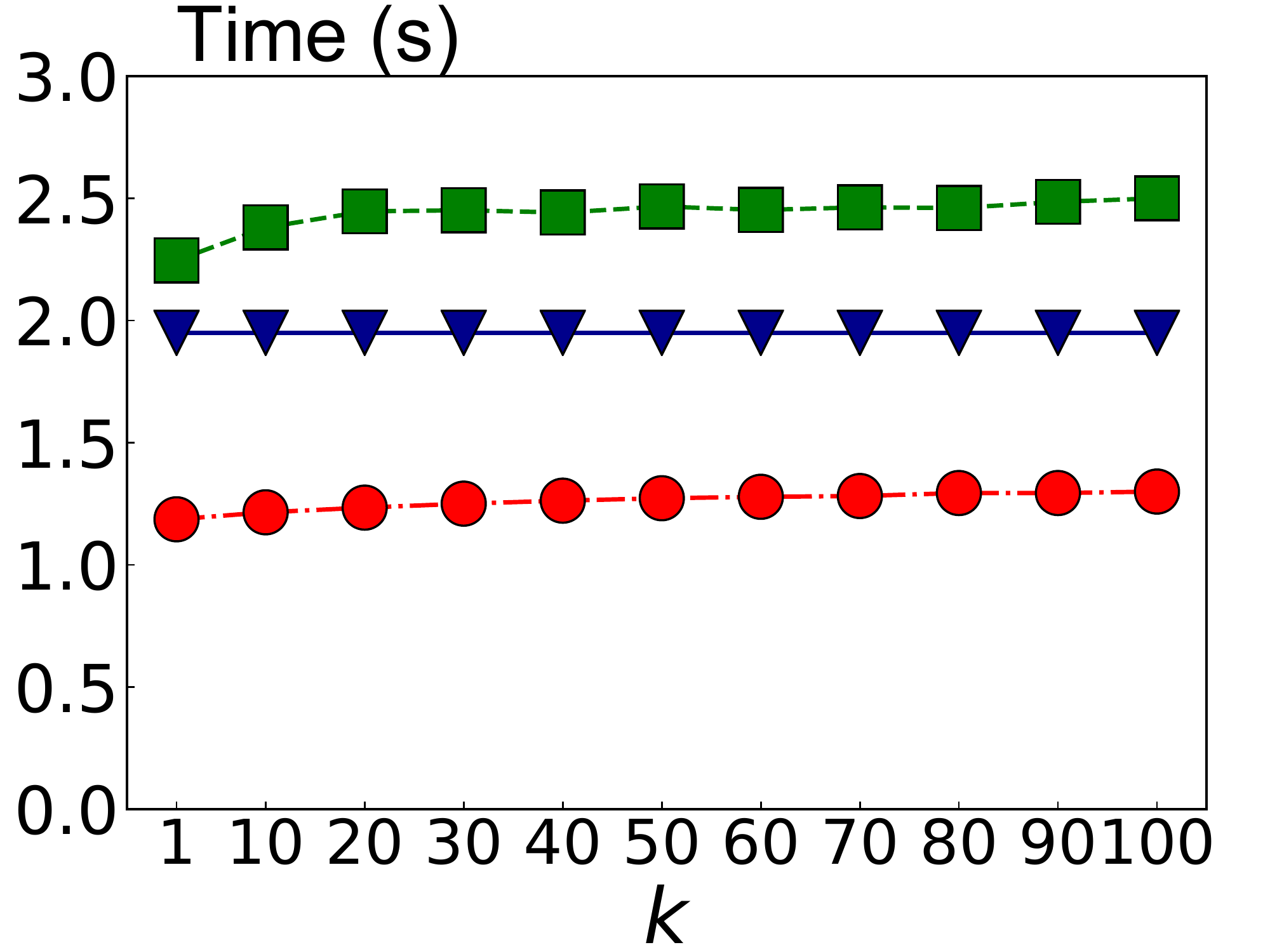}
		\end{minipage}
	}
	\subfigure[T-drive with Frechet]{
		\begin{minipage}[c]{0.45\linewidth}
			\centering
			\includegraphics[width=1\textwidth]{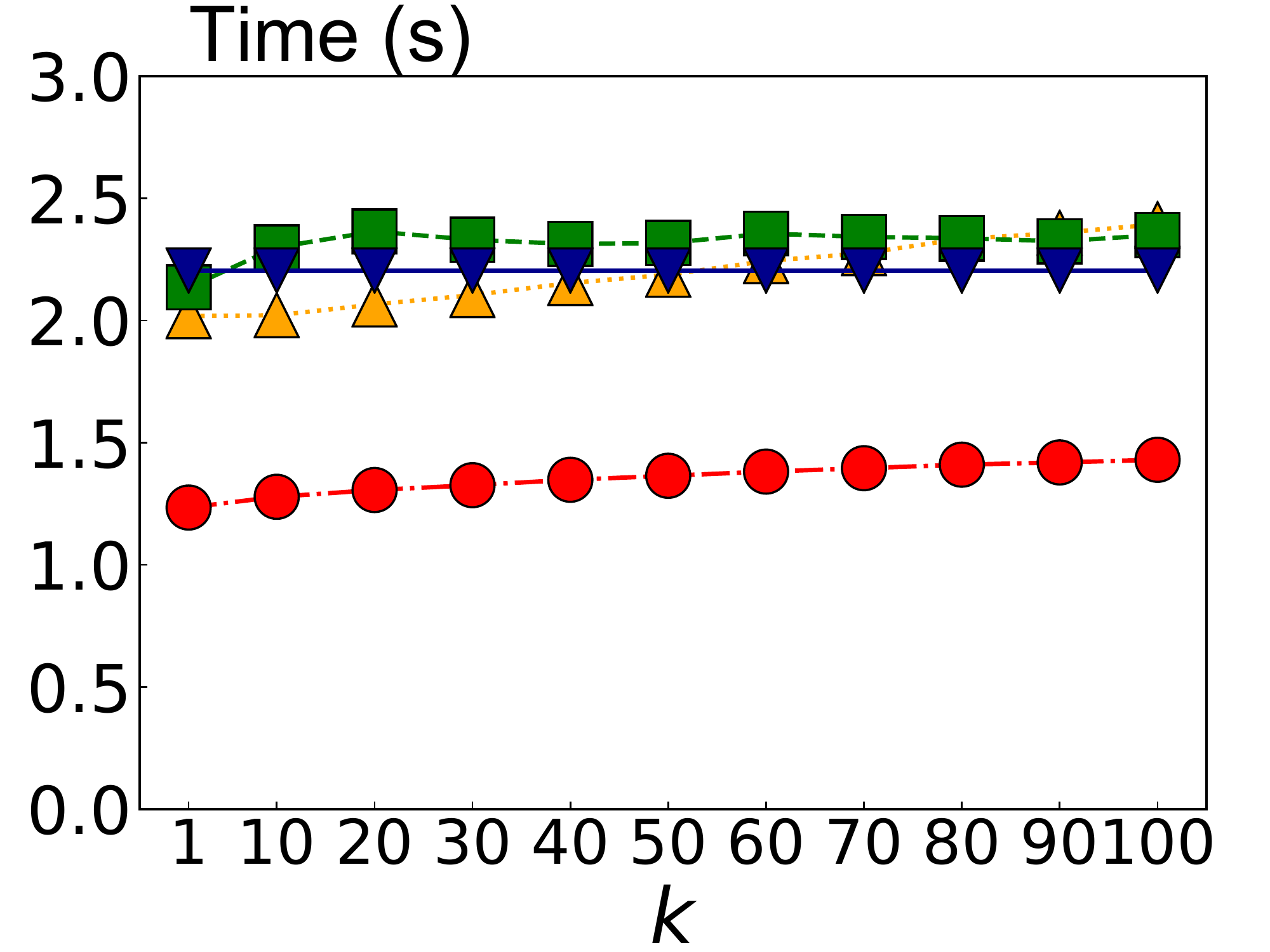}
		\end{minipage}
	}
	\subfigure[Xi'an with Hausdorff]{
		\begin{minipage}[c]{0.45\linewidth}
			\centering
			\includegraphics[width=1\textwidth]{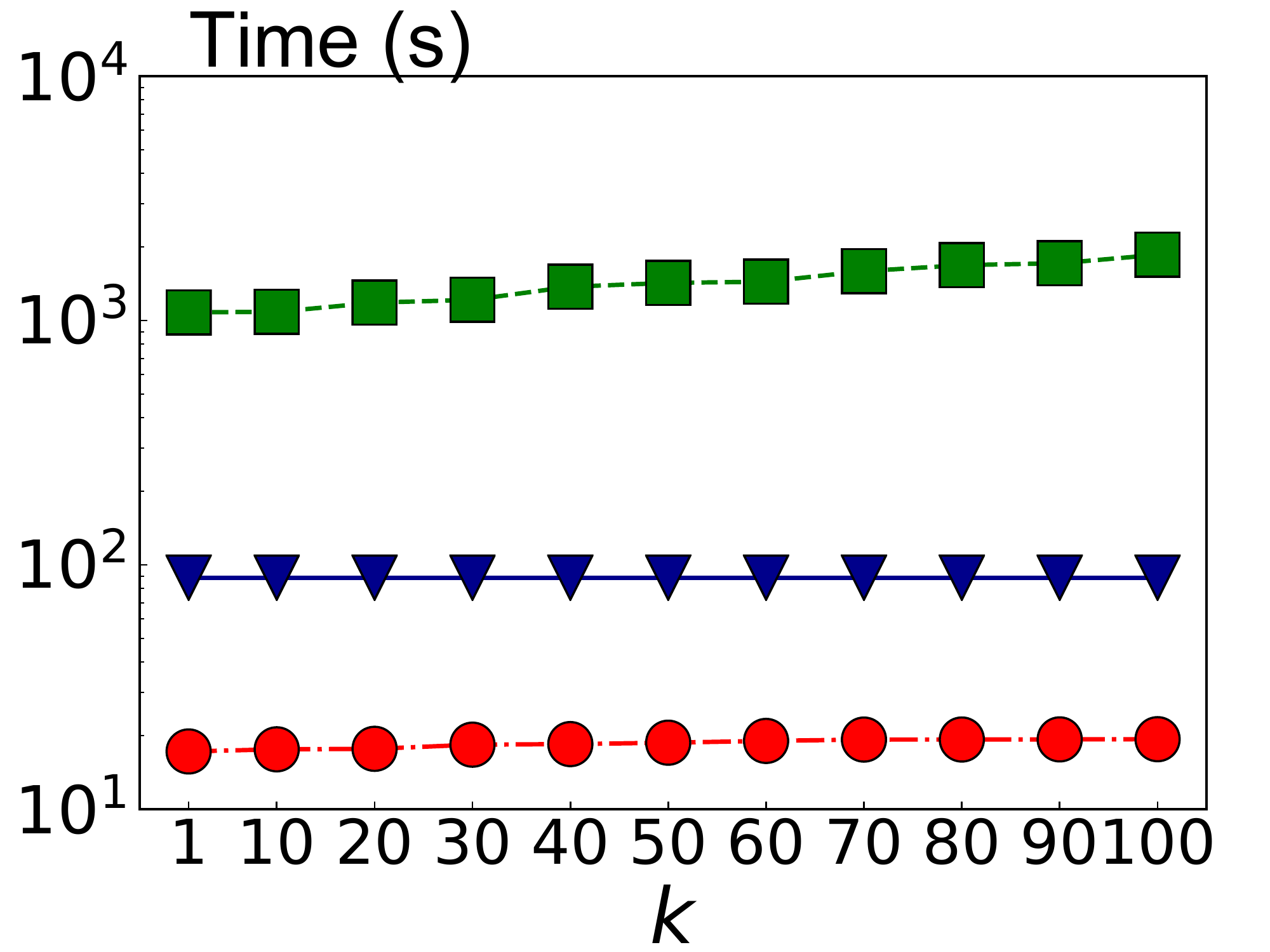}
		\end{minipage}
	}
	\subfigure[Xi'an with Frechet]{
		\begin{minipage}[c]{0.45\linewidth}
			\centering
			\includegraphics[width=1\textwidth]{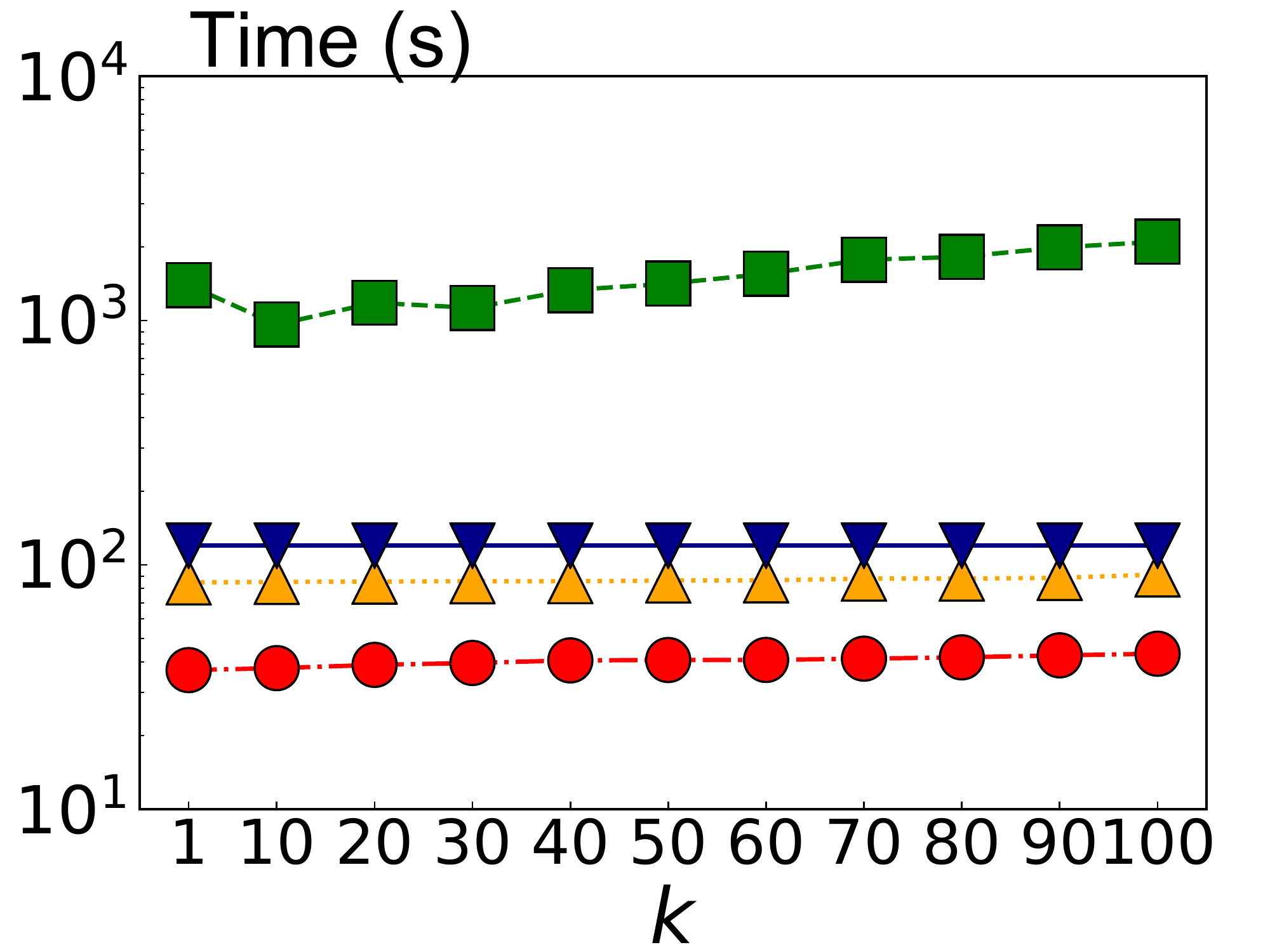}
		\end{minipage}
	}
	\subfigure[Osm with Hausdorff]{
		\begin{minipage}[c]{0.45\linewidth}
			\centering
			\includegraphics[width=1\textwidth]{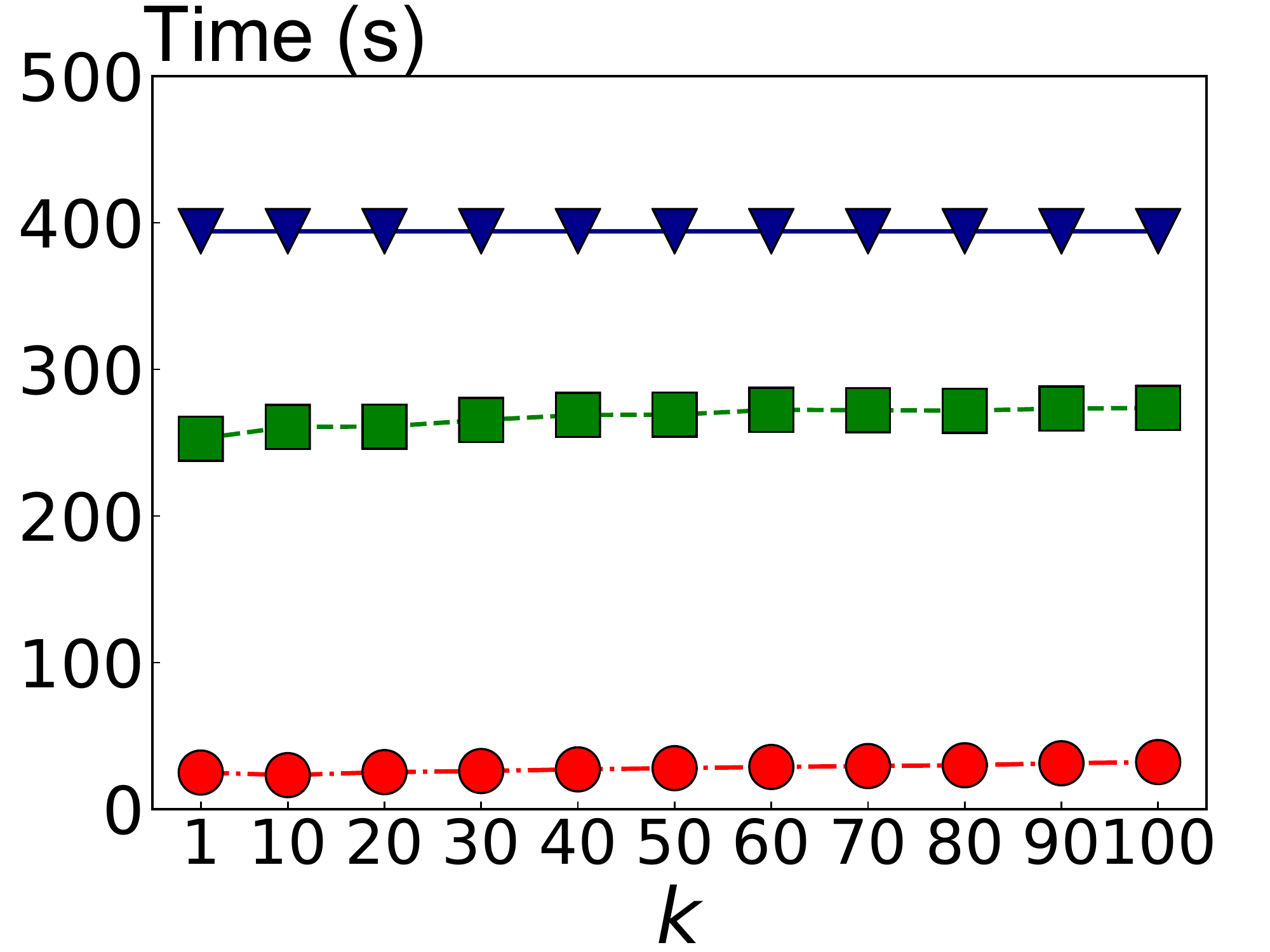}
		\end{minipage}
	}
	\subfigure[Osm with Frechet]{
		\begin{minipage}[c]{0.45\linewidth}
			\centering
			\includegraphics[width=1\textwidth]{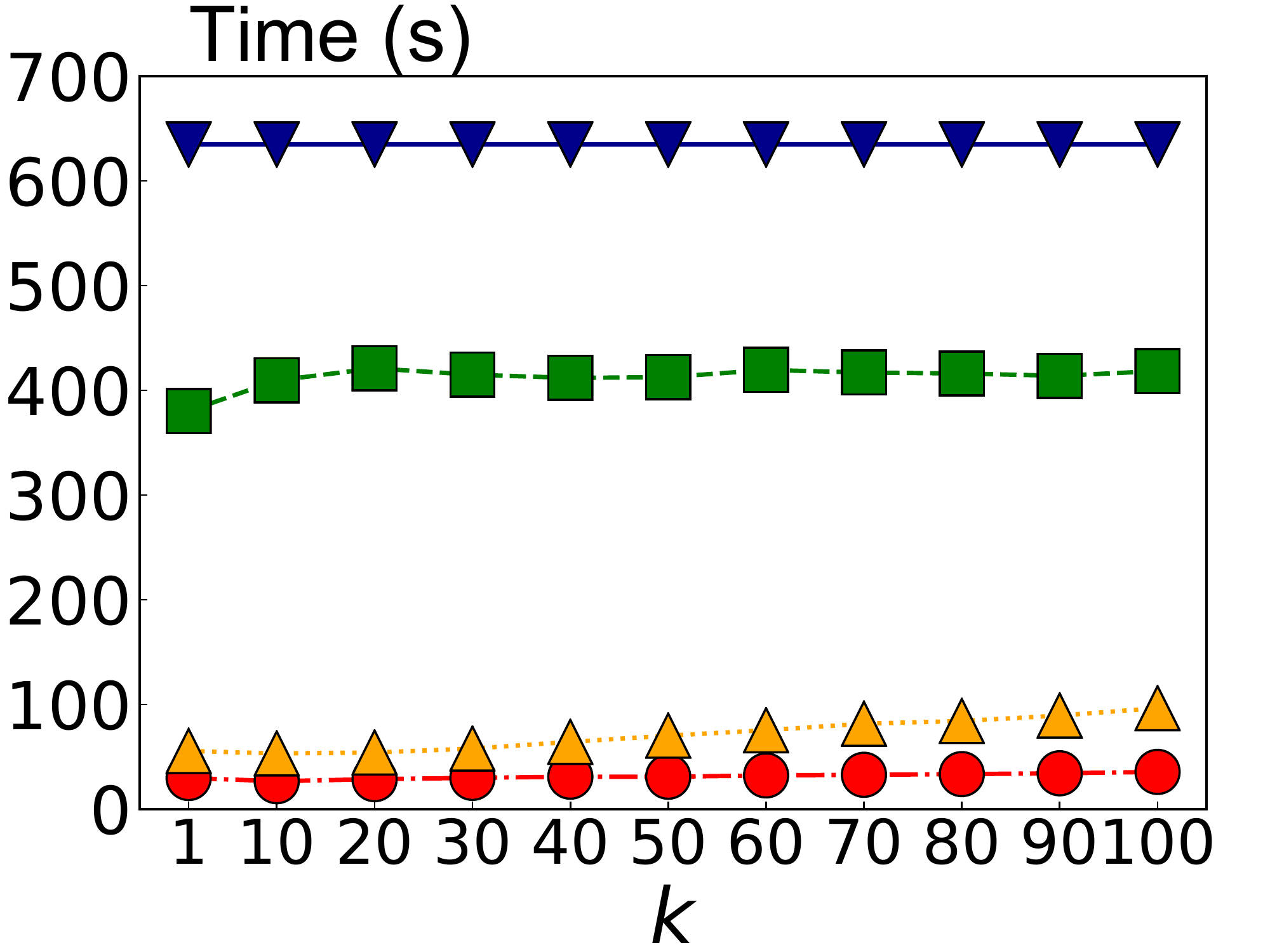}
		\end{minipage}
		\label{fig:k_osm_frechet}
	}
	\caption{Performance when Varying $k$}
	\label{knnfig}
\end{figure}

\textbf{Varying the value of $k$}. 
Fig. \ref{knnfig} shows the performances of all algorithms when varying the value of $ k $ in $ \{1,10,20,\dots,100\} $. We notice that with the increase of $k$, the changing trends differ across the algorithms. For example, the query time of DFT is unstable since it finds $ C * k $ trajectories at random from the dataset and uses the $ k $-th smallest distance as the threshold. The query time depends on the quality of the threshold.
The query time of DITA increases when $ k $ increases. DITA repeatedly reduces its threshold by half until the number of candidate trajectories is less than $ C * k $. The $ k $-th smallest distance is used to perform a range search.
LS computes the distance between the query trajectory and all trajectories, so the value of $k$ has little effect.
\textsf{REPOSE} uses the current best $k$-th result as a pruning threshold, so the query time also increases when $ k $ increases.
\textsf{REPOSE} achieves the best performance for all $ k $ and offers a considerable improvement. 
%In Fig. \ref{fig:k_osm_frechet}, for the Frechet distance on OSM, the query times are as follows: \textsf{REPOSE} is 35.6s, DITA is 94.1s, DFT is 418.3s, and LS is 634.7s.

\begin{table}[t]
	\centering
	\caption{Query time of $\delta$ on $ D_H $ (Hausdorff) and $ D_F $ (Frechet)}
	\label{tb:parameter_expe_delta}
	\resizebox{0.43\textwidth}{!}{
	\begin{tabular}{|m{0.64cm}<{\centering}|m{0.5cm}<{\centering}|m{0.5cm}<{\centering}|m{0.64cm}<{\centering}|m{0.5cm}<{\centering}|m{0.5cm}<{\centering}|m{0.64cm}<{\centering}|m{0.5cm}<{\centering}|m{0.5cm}<{\centering}|}
		\hline
		\multicolumn{3}{|c|}{T-drive} & \multicolumn{3}{c|}{Xi'an} & \multicolumn{3}{c|}{OSM}\\
		\hline
		\multirow{2}*{Value}  & \multicolumn{2}{c|}{QT (s)}  &\multirow{2}*{Value}& \multicolumn{2}{c|}{QT (s)} & \multirow{2}*{Value} & \multicolumn{2}{c|}{QT (s)} \\
		\cline{2-3}
		\cline{5-6}
		\cline{8-9}
		& $ D_H  $ &  $ D_F$ &  &   $ D_H $ &  $  D_F $ &  &  $ D_H $ &  $  D_F $ \\
		\hline
		0.01 &5.04 &4.39 &0.005 &28.52 &68.60 &0.1&50.58&50.99 \\	
		0.05&4.06 &3.19 &\textbf{0.010} &\textbf{19.14} &61.31 &0.5&36.70 &39.52\\
		0.10&2.47 &2.55 &0.015 &21.88 &54.11 &\textbf{1.0}&\textbf{32.16} &\textbf{35.61}\\
		\textbf{0.15}&\textbf{1.29} &\textbf{1.28} &0.020 &23.98 &48.46 &1.5&35.66 &37.71\\
		0.20&2.26 &1.41 &0.025 &29.81 &51.76 &2.0&37.10 &38.46\\
		0.25&2.50 &1.50 &\textbf{0.030} &32.80 &\textbf{43.27} &2.5&39.23 &41.55\\
		0.30&3.04 &2.60 &0.035 &35.10 &56.34 &3.0&43.39 &45.54\\
		\hline
	\end{tabular}}
\end{table}

\textbf{Parameter chosen on $ \delta $}. 
We study the effect of $ \delta $ on \textsf{REPOSE}. With the change of $ \delta $, there are descending and ascending trends on the query time. The reason for the descending trend is that reference trajectories become longer quickly when $ \delta $ decreases. Although a long reference trajectory can improve the pruning, the benefits may become small after exceeding a certain length, since  time overheads for computing $LB_o$ and $LB_t$ are introduced.
The ascending trend occurs because the number of cells decreases when $ \delta $ increases. Then a reference trajectory cannot approximate a trajectory well, which results in a poor pruning and larger query times.
In summary, the value of $ \delta $ affects the query time significantly,  and it is important to select an appropriate $ \delta $ value. Table \ref {tb:parameter_expe_delta} shows that the effect of $ \delta $ is relatively consistent on the same dataset. Therefore, we can reuse the setting of $ \delta $ from one measure to another when the query efficiency meets the application requirements.

\begin{table}[t]
	\centering
	\caption{Query time of $N_p$ on $ D_H $ (Hausdorff) and $ D_F $ (Frechet)}
	%\vspace*{0.1in}
	\label{tb:parameter_expe_pn}
	\resizebox{0.43\textwidth}{!}{
	\begin{tabular}{|m{0.64cm}<{\centering}|m{0.5cm}<{\centering}|m{0.5cm}<{\centering}|m{0.64cm}<{\centering}|m{0.5cm}<{\centering}|m{0.5cm}<{\centering}|m{0.64cm}<{\centering}|m{0.5cm}<{\centering}|m{0.5cm}<{\centering}|}
		\hline
		\multicolumn{3}{|c|}{T-drive} & \multicolumn{3}{c|}{Xi'an} & \multicolumn{3}{c|}{OSM}\\
		\hline
		\multirow{2}*{Value}  & \multicolumn{2}{c|}{QT (s)}  &\multirow{2}*{Value}& \multicolumn{2}{c|}{QT (s)} & \multirow{2}*{Value} & \multicolumn{2}{c|}{QT (s)} \\
		\cline{2-3}
		\cline{5-6}
		\cline{8-9}
		& $ D_H  $ &  $ D_F$ &  &   $ D_H $ &  $  D_F $ &  &  $ D_H $ &  $  D_F $ \\
		\hline
		1 &1.45 &1.49 &1 &22.20 &49.14 &1&35.09&39.05 	\\	
		\textbf{3} &\textbf{1.21} &1.34 &3&21.10 &46.89 &3&34.36&37.85		\\
		\textbf{5} &1.29 &\textbf{1.28} &\textbf{5} &\textbf{19.14} &43.27 &\textbf{5}&\textbf{32.16}&\textbf{35.61}       \\
		7 &1.30 &1.51 &\textbf{7} &22.99 &\textbf{41.41} &7&32.33&35.85      \\
		9 &1.40 &2.02 &9 &23.38 &48.20 &9&39.43&36.14       \\
		11 &1.55 &2.43 &11 &24.07 &49.91 &11&38.86& 40.02     \\
		\hline
	\end{tabular}}
\end{table}

\textbf{Parameter chosen on $ N_p $}. 
Table \ref {tb:parameter_expe_pn} shows the effect of $ N_p $ on the algorithms. Similar to $ \delta $, there are descending and ascending trends on query time.  As $ N_p $ increases beyond a boundary value, the improvement on pruning is small and additional computational overhead is introduced.

\begin{figure}[t]
	\centering
	\includegraphics[width=.4\textwidth]{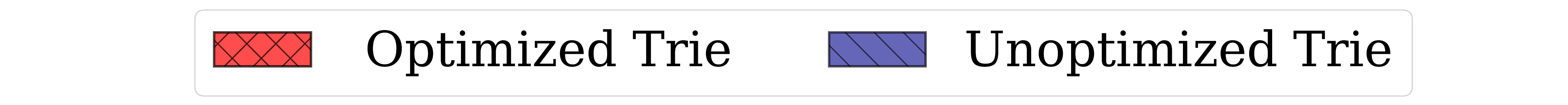} %1.pdf是图片文件的相对路径	
	\subfigure[Reduced \# of trie nodes]{
		\begin{minipage}[c]{0.45\linewidth}
			\centering
			\includegraphics[width=1\textwidth]{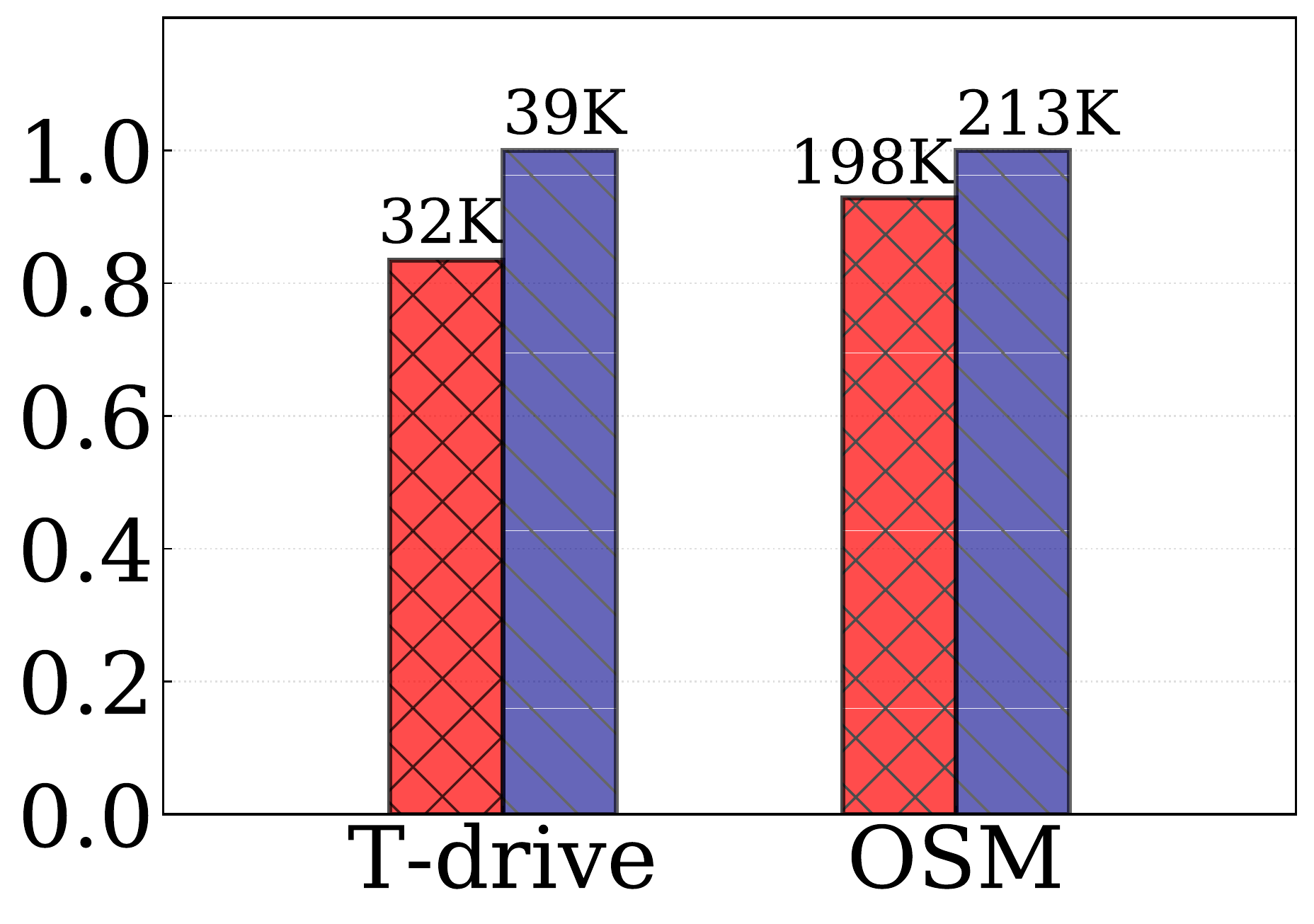}
		\end{minipage}
	}
	\subfigure[Improvment on query time]{
		\begin{minipage}[c]{0.45\linewidth}
			\centering
			\includegraphics[width=1\textwidth]{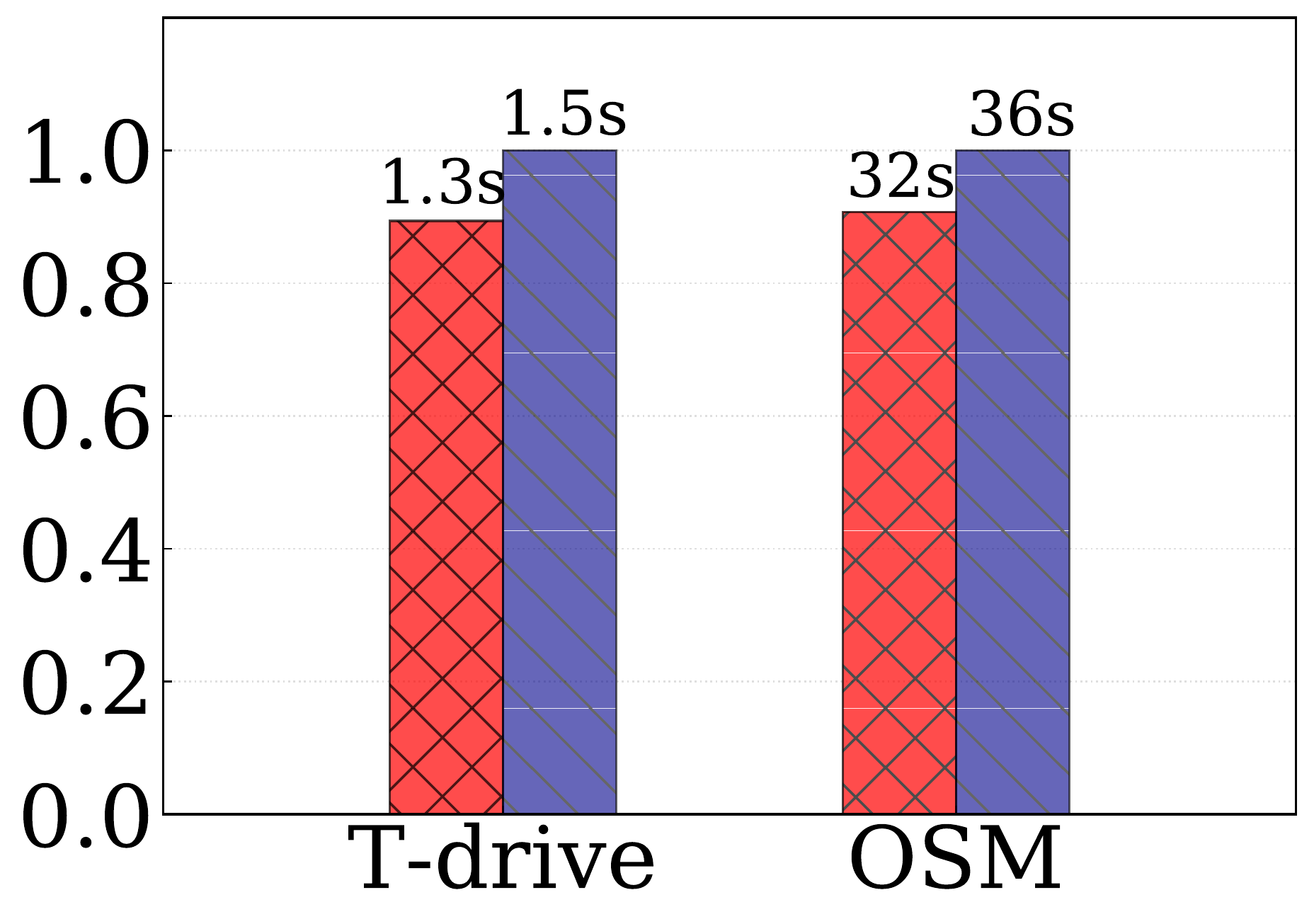}
		\end{minipage}
	}
	\caption{Improvement by Optimized Trie}
	\label{triefig}
\end{figure}

\begin{figure}[t]
	\centering
	\includegraphics[width=.4\textwidth]{fig/knn/legend.pdf} %1.pdf是图片文件的相对路径
	\subfigure[Osm with Hausdorff]{
		\begin{minipage}[c]{0.45\linewidth}
			\centering
			\includegraphics[width=1\textwidth]{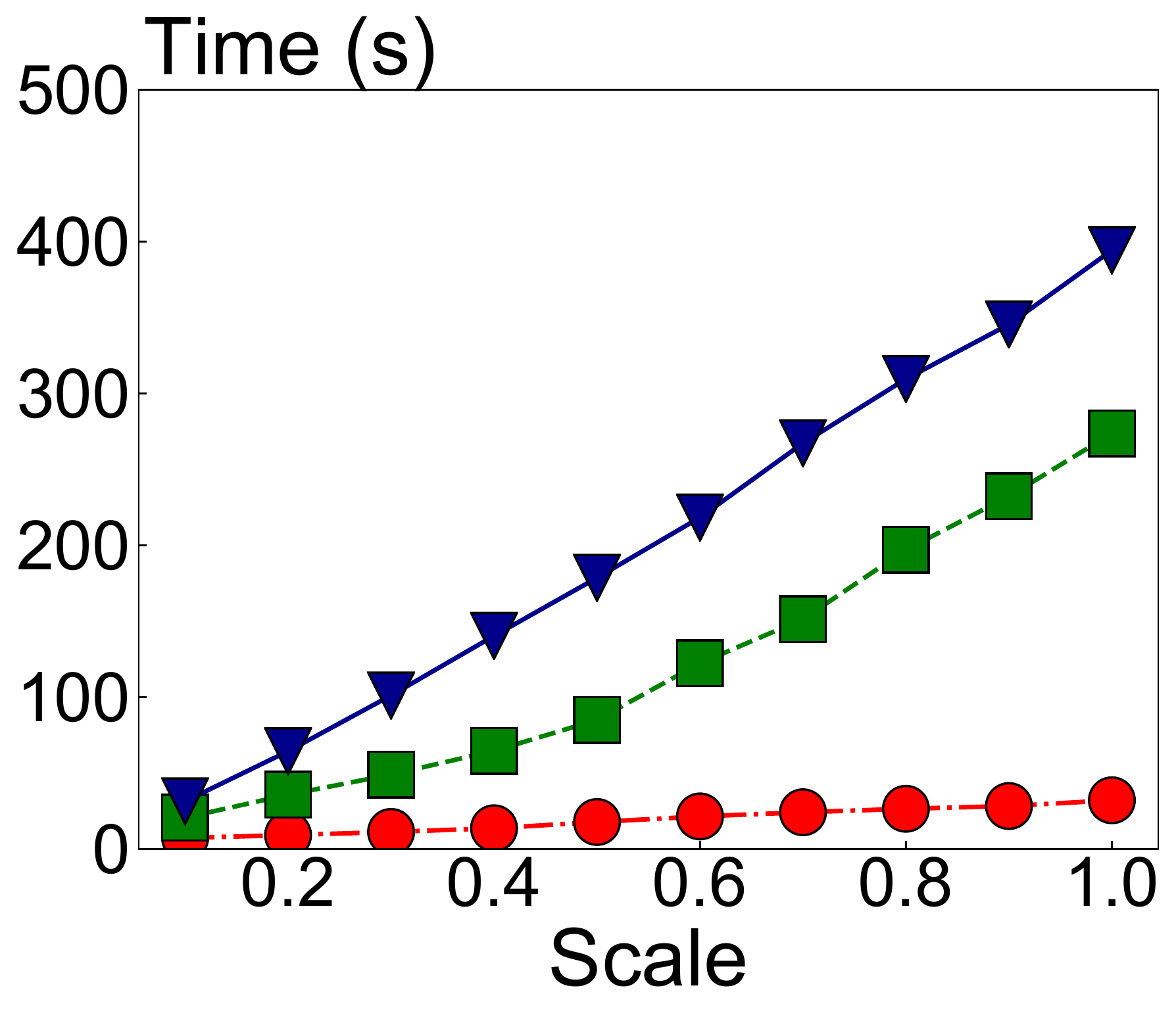}
		\end{minipage}
	}	
	\subfigure[Osm with Frechet]{
		\begin{minipage}[c]{0.45\linewidth}
			\centering
			\includegraphics[width=1\textwidth]{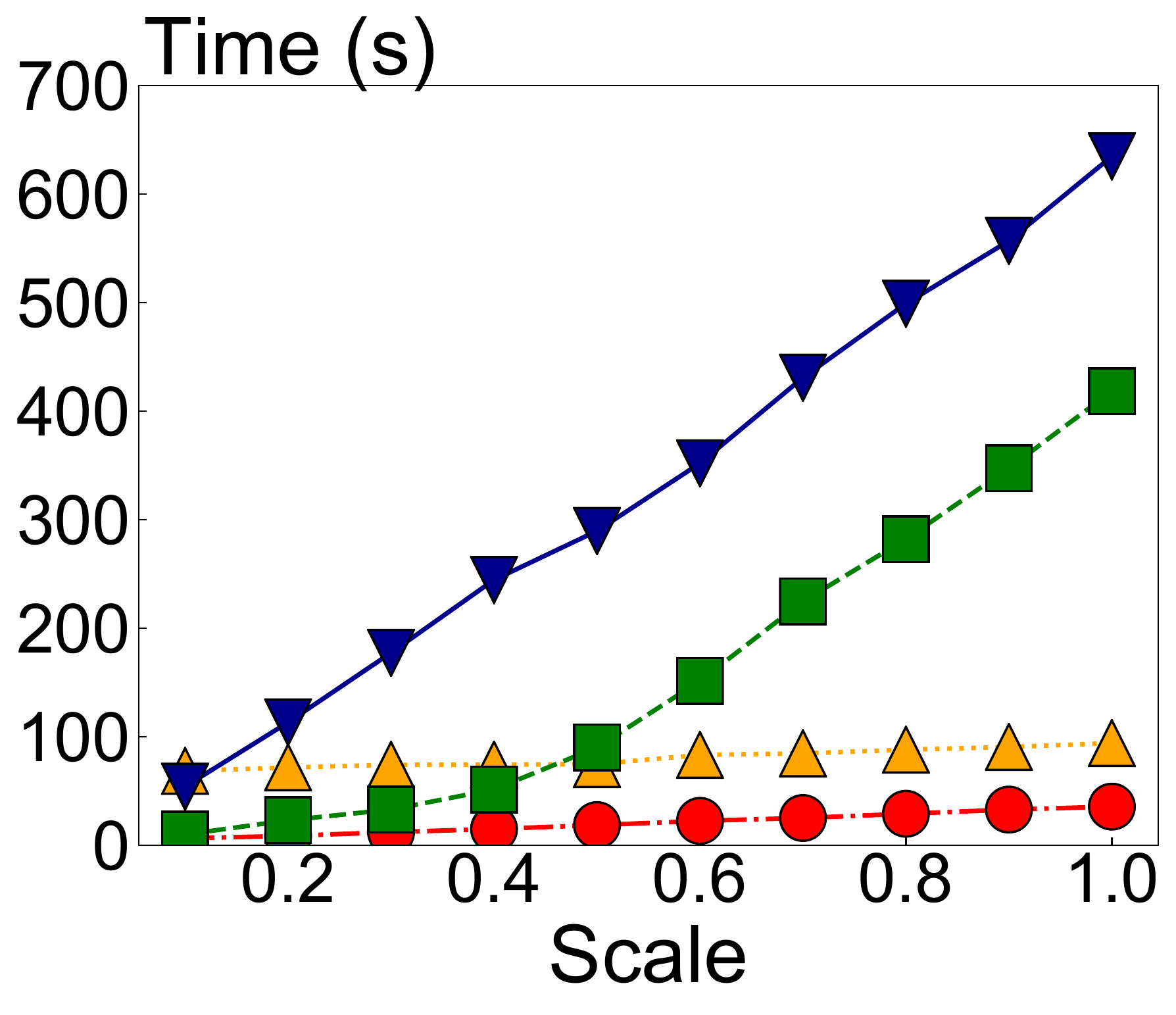}
		\end{minipage}
	}
	\caption{Effect of Dataset Cardinality}
	\label{fig:sizetime}
\end{figure}

\begin{figure}[t]
	\centering
	\subfigure[Osm with Hausdorff]{
		\begin{minipage}[c]{0.45\linewidth}
			\centering
			\includegraphics[width=1\textwidth]{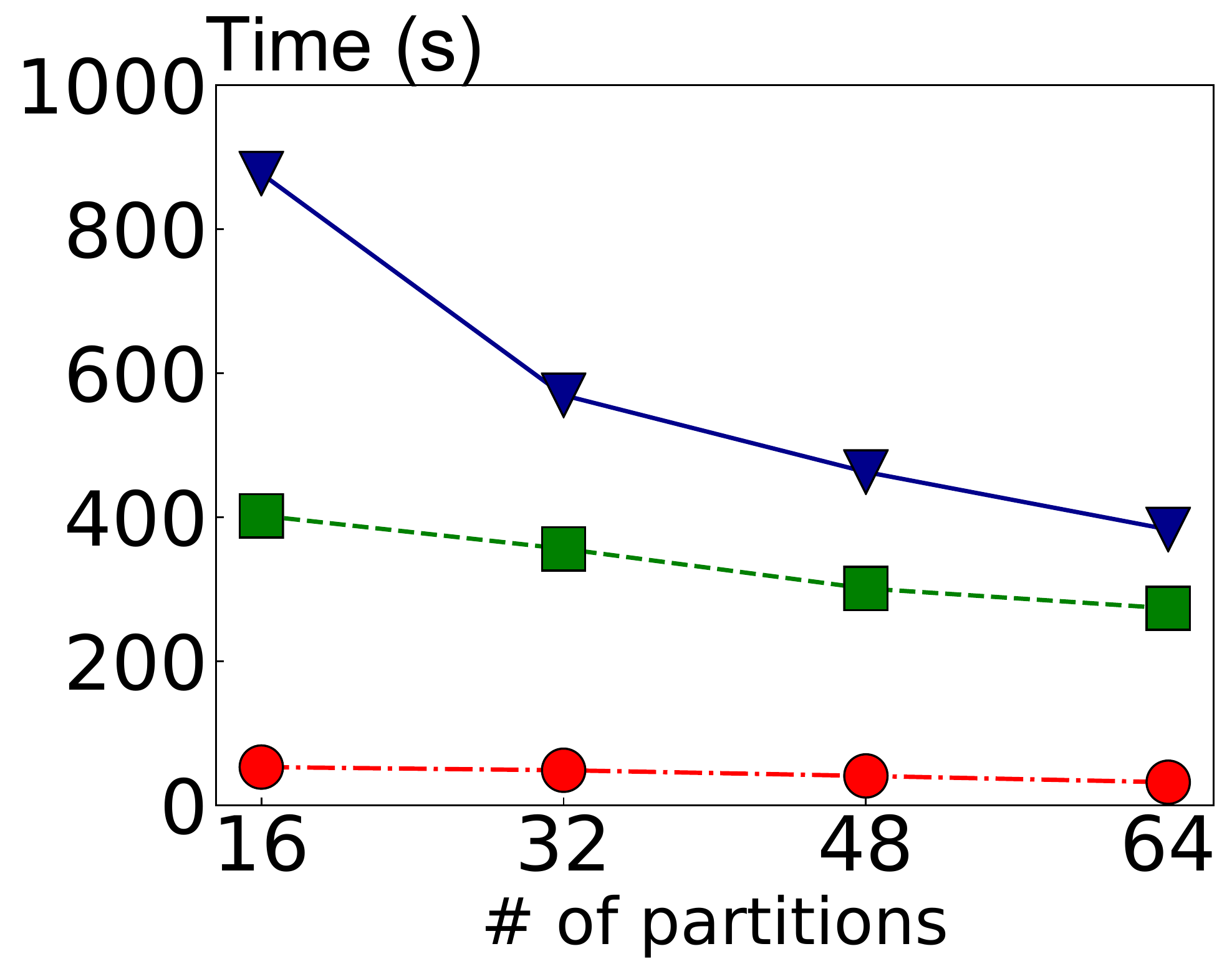}
		\end{minipage}
	}	
	\subfigure[Osm with Frechet]{
		\begin{minipage}[c]{0.45\linewidth}
			\centering
			\includegraphics[width=1\textwidth]{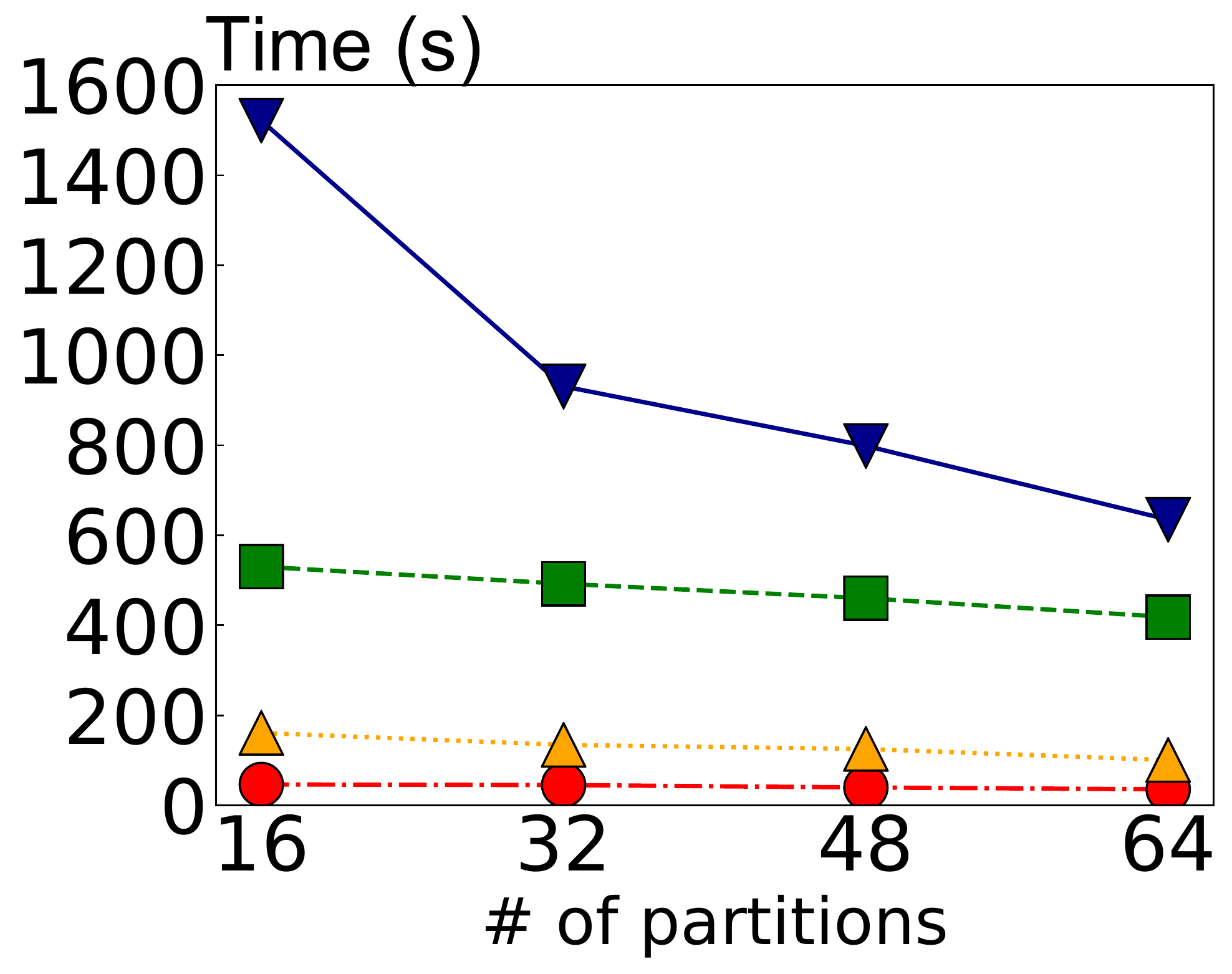}
		\end{minipage}
	}
	
	\caption{Effect of the Number of Partitions}
	\label{fig:coretime}
\end{figure}

\textbf{Improvement by optimized trie.}
Fig. \ref{triefig} shows the improvements by using the optimized trie on the query time and the reduced number of trie nodes. For T-drive dataset, the number of trie nodes is reduced by about 20\% compared to that of using unoptimized trie. The query time is decreased by about 12\%. For OSM dataset, both of the number of trie nodes and the query time are reduced by about 8\%.

\textbf{Effect of dataset cardinality.}
From Fig. \ref{fig:sizetime}, we observe that the query time of all algorithms grows linearly. \textsf{REPOSE} has the best performance due to a better partitioning scheme and a powerful pruning.

\begin{table}[t]
	\centering
	\caption{Effect of Partitioning Strategy}
	%\vspace*{0.1in}
	\label{tb:partition_strategy}
	\scriptsize
	%\footnotesize
	\begin{tabular}{|m{0.95cm}<{\centering}|m{1.55cm}<{\centering}|m{1.2cm}<{\centering}|m{0.95cm}<{\centering}|m{0.9cm}<{\centering}|}
		\hline
		\textbf{Distance}&\textbf{Partitioning}&\textbf{T-drive (s)} &\textbf{Xi'an (s)}&\textbf{OSM (s)} \\
		\hline
		\hline
		\multirow{3}*{Hausdorff}&Heterogeneous&\textbf{1.29} &\textbf{19.14} &\textbf{32.16}\\			
		\cline{2-5}	
		& Homogeneous &1.64 &35.93 &344.74\\
		\cline{2-5}
		& Random &1.51 &21.38 &35.12 \\
		\cline{1-5}	
		\multirow{3}*{Frechet}&Heterogeneous&\textbf{1.28} &\textbf{43.27} &\textbf{35.61}\\			
		\cline{2-5}	
		&Homogeneous&1.52 &109.42 &240.32\\
		\cline{2-5}
		&Random&1.45 &47.94 &37.68 \\
		\cline{1-5}		
		\hline	
	\end{tabular}
\end{table}

\textbf{Effect of the number of partitions.}
Fig. \ref{fig:coretime} shows that 
with the increase of the number of partitions from 16 to 64, all algorithms gain performance improvements. 
\cmt{We have 64 partitions by default, where each core processes a partition.}
\textsf{REPOSE} achieves a higher performance gain (not including LS) because it implements a better partitioning scheme by equalizing the workload of each partition. LS has the highest performance gain. This is because it suffers from a severe data skew issue when the number of partitions is small. When we increase the number of partitions, the performance is significantly improved.

%\begin{figure}[htbp]
%	\centering
%	\includegraphics[width=.4\textwidth]{fig/3PartitionRepose/3PartitionReposeLegend.pdf} %1.pdf是图片文件的相对路径	
%	\subfigure[Query time ratio on Hausdorff]{
%		\begin{minipage}[c]{0.45\linewidth}
%			\centering
%			\includegraphics[width=1\textwidth]{fig/3PartitionRepose/Hausdorff_3PartitionRepose.pdf}
%		\end{minipage}
%	}
%	\subfigure[Query time ratio on Frechet]{
%		\begin{minipage}[c]{0.45\linewidth}
%			\centering
%			\includegraphics[width=1\textwidth]{fig/3PartitionRepose/Frechet_3PartitionRepose.pdf}
%		\end{minipage}
%	}
%	
%	\caption{Effect on three partition strategies}
%	\label{fig:3PartitionRepose}
%\end{figure}

\cmt{\textbf{Effect of partitioning strategies.} 
We deploy three different partitioning strategies, heterogeneous, homogeneous, and random, with the RP-Trie as a local index to analyze the effects of the partitioning strategies.
In the homogeneous strategy, we place similar trajectories in a cluster in the same partition.
In the random strategy, we place trajectories into partitions at random.
The results are shown in Table \ref{tb:partition_strategy}.
The heterogeneous partitioning strategy achieves the best performance, and the homogeneous strategy has the worst performance, which can be explained as follows:
(1) The local pruning of the homogeneous strategy is less effective since the distances of the trajectories in a partition to the query trajectory are very similar.
(2) In the heterogeneous strategy, a high degree of data discrimination in a partition is able to improve the local pruning. In addition, since the composition of each partition is similar, the query times of all partitions are similar, which enables load-balancing.}

\renewcommand\arraystretch{1.3}
\begin{table}[t]
	\centering
	\caption{Comparison with DITA using Heterogeneous Partitioning}
	%\vspace*{0.1in}
	\label{tb:heter_dita}
	\scriptsize
	\begin{tabular}{|m{0.85cm}<{\centering}|m{1.5cm}<{\centering}|m{1.2cm}<{\centering}|m{0.95cm}<{\centering}|m{0.9cm}<{\centering}|}
		\hline
		\textbf{Distance}&\textbf{Algorithm}&\textbf{T-drive (s)} &\textbf{Xi'an (s)}&\textbf{OSM (s)} \\
		\hline
		\hline
		\multirow{3}*{DTW}	&REPOSE&\textbf{1.48} &\textbf{180.09} &\textbf{499.52} \\
		\cline{2-5}	
		&Heter-DITA&2.74 &184.97 &503.59\\
		\cline{2-5}
		&DITA&2.78 &186.36 &509.53\\			
		\cline{1-5}	
		\multirow{3}*{Frechet}&REPOSE&\textbf{1.28} &\textbf{43.27} &\textbf{35.61} \\
		\cline{2-5}
		&Heter-DITA&2.26 &84.94 &87.56\\
		\cline{2-5}
		
		&DITA&2.39 &91.31 &94.12\\
		\cline{1-5}		
		\hline	
	\end{tabular}
\end{table}

\begin{table}[t]
	\centering
	\caption{Comparison with DFT using Heterogeneous Partitioning}
	\label{tb:heter_dft}
	\scriptsize
	%\footnotesize
	\begin{tabular}{|m{1cm}<{\centering}|m{1.5cm}<{\centering}|m{1.2cm}<{\centering}|m{0.95cm}<{\centering}|m{0.9cm}<{\centering}|}
		\hline
		\textbf{Distance}&\textbf{Algorithm}&\textbf{T-drive (s)} &\textbf{Xi'an (s)}&\textbf{OSM (s)} \\
		\hline
		\hline
		\multirow{3}*{Hausdorff}&REPOSE&\textbf{1.29} &\textbf{19.14} &\textbf{32.16} \\	
		\cline{2-5}	
		&Heter-DFT&2.32 &1645.59 &260.77\\
		\cline{2-5}
		&DFT&2.50 &1857.88&273.68\\	
		\cline{1-5}	
		\multirow{3}*{Frechet}&REPOSE&\textbf{1.28} &\textbf{43.27} &\textbf{35.61} \\			
		\cline{2-5}	
		&Heter-DFT&2.21 &1855.01 &380.47\\
		\cline{2-5}
		&DFT&2.33 &2100.80 &418.40\\	
		\cline{1-5}		
		\hline	
	\end{tabular}
\end{table}

\cmt{\textbf{Comparison with DITA and DFT using heterogeneous partitioning.}
To further explore the benefits of the heterogeneous partitioning strategy, we apply it to DITA and DFT to examine the effects on performance. We denote these two methods as Heter-DITA and Heter-DFT. The results are shown in Tables \ref{tb:heter_dita} and \ref{tb:heter_dft}. Since DITA does not support Hausdorff, we examine its performance on DTW instead. The results show that both Heter-DITA and Heter-DFT outperform the original DITA and DFT, which offers evidence that the heterogeneous partitioning strategy is superior.}

%% file: related.tex
\section{Related Work}\label{sec:related}

Trajectory similarity search is a fundamental operation in many applications \cite{DBLP:conf/time/DingTS08,DBLP:conf/icde/FrentzosGT07,DBLP:conf/gis/BakalovHT05,DBLP:conf/mdm/BakalovHKT05,DBLP:conf/icde/VlachosGK02,DBLP:journals/pvldb/XieLP17,DBLP:conf/compgeom/DriemelS17,DBLP:conf/vldb/Keogh02,DBLP:conf/sigmod/FaloutsosRM94,DBLP:conf/kdd/VlachosHGK03}. Based on the index structures, we classify the existing proposals into three categories:

\textbf{(1) Space partitioning tree based approach}. Traditional space partitioning trees \cite{DBLP:conf/vldb/Keogh02,DBLP:conf/sigmod/FaloutsosRM94,DBLP:conf/kdd/VlachosHGK03,DBLP:conf/kdd/WangZX14} (e.g., R-tree) construct MBRs by dividing trajectories into sub-trajectories and prune non-candidate trajectories based on the MBR intersections. However, large MBR areas limit the pruning capabilities. For metric spaces, classical metric indexes, such as the VP-tree and the M-tree \cite{DBLP:journals/vldb/FuCCM00,DBLP:conf/soda/Yianilos93,DBLP:journals/ipl/Uhlmann91}, are used for Hausdorff and Frechet. However, these methods perform poorly in practice since they abandon the features of the trajectory and use only the metric characteristics to index the trajectories.

\textbf{(2) Grid hashing based approach}. The grid hashing based approaches \cite{DBLP:conf/compgeom/DriemelS17,DBLP:conf/gis/AstefanoaeiCKGS18} use the grid sequence that the trajectories traverse to represent the trajectories, which captures the trajectory characteristics and makes it possible to quickly find candidate trajectories. Although these methods fail to return exact results, they demonstrate good performances in approximate trajectory similarity search.

\textbf{(3) Distributed index based approach}.
\cmt{SmartTrace$^+$ \cite{DBLP:journals/tkde/Zeinalipour-YaztiLCVAG13} is a real-time distributed trajectory similarity search framework for smartphones that employs a similar pruning and refinement paradigm of \textsf{REPOSE}.}
DFT \cite{DBLP:journals/pvldb/XieLP17} uses the R-tree to index trajectory line segments and prunes by using the distances between a query trajectory and the MBRs of line segments. Unfortunately, DFT needs to reconstruct the line segments into trajectories when computing similarity, which incurs high space overhead.
DITA \cite{DBLP:conf/sigmod/Shang0B18} selects pivot points to represent a trajectory and stores close pivot points in the same MBR. Then, DITA utilizes the trie to index all MBRs and prunes by using the distances between the query trajectory and the MBRs. However, DITA fails to retain the features of original trajectories due to its pivot point selection strategy. Moreover, the query processing is inefficient for some distance metrics.

%% file: conclusion.tex
\section{Conclusion and Future Work}\label{sec:conclusion}
We present a distributed in-memory framework, \textsf{REPOSE}, for processing top-$k$ trajectory similarity queries on Spark.
%The framework supports most of the widely used trajectory similarity measures, including Hausdorff, Frechet, DTW, EDR, ERP, and LCSS. 
We develop a heterogeneous global partitioning strategy that places similar trajectories in different partitions in order to achieve load balancing and to utilize all compute nodes during query processing. We design an effective index structure, RP-Trie, for local search that utilizes two optimizations, namely a succinct trie structure and z-value re-arrangement. To further accelerate the query processing, we provide several pruning techniques. The experimental study uses 7 widely used datasets  and offers evidence that \textsf{REPOSE} outperforms its competitors in terms of query time, index size and index construction time. Compared to the state-of-the-art proposals for Hausdorff, \textsf{REPOSE} improves the query time by about 60$\%$, the index size by about 80$\%$, and the index construction time by about 85$\%$.

\cmt{In the future work, it is of interest to take the temporal dimension into account to enable top-$k$ spatial-temporal trajectory similarity search in distributed settings. In addition, it is of interest to develop optimizations for other distance measures, such as LCSS, ERP, and EDR, when using the RP-Trie index.}

%% file: appendix.tex
\newpage

\appendix

\subsection{Pseudocode of Search Procedure}

The search procedure is detailed in Algorithm \ref{LocalSearch}. 

\begin{algorithm}
	\caption{\textsf{LocalSearch}($ \tau_q,root,k$)}  
	\label{LocalSearch} 
	\LinesNumbered
	\KwIn{The query trajectory $ \tau_q $, the $root$ node, and $k$}
	\KwOut{The top-$ k $ results}	
	$ d_k \gets +\infty$; $ minHeap \gets \emptyset$\;
	Initialize a priority queue $E$ with $root$\;
	\While{$ E \ne \emptyset $}{
		$ t\gets E.top $\;
		\If{$t$ is a leaf node}
		{
			\If{$t.LB_t< d_k \land t.LB_p< d_k$}{	
				\For{each $ \tau \in t $}
				{
					Compute $\textsf{D}_\textsf{H}$($\tau_q$,$\tau$)\;
					Update $ minHeap $ and $d_k$ if necessary\;
				}
			}
			\Else{Break\;}
		}
		\Else{
			\If{$t.LB_o< d_k \land t.LB_p< d_k$}{	
				\For{each child node $ t'$ of $t $} 
				{
					$p^*\gets$ the reference point of $t'$\; 
					$t' \gets$ \textsf{CompLB}($\tau_q,p^*,t.r,t.c_{max}$)\;
					Insert $t'$ into $E$\;	
				}
			}
			\Else{Break\;}
		}
	}
	\Return $k$ trajectories in $ minHeap $\;
\end{algorithm}

\subsection{Details of the Greedy Algorithm for the Optimized RP-Trie}

Assume that we have a set of $N$ reference trajectories $\mathcal{Z}=\{Z_1,\dots,Z_N\}$ and $M$ cells. Let $L$ be the maximum length of reference trajectories and $\mathcal{L}(\mathcal{Z})=\sum_{i=1}^{N} |Z_i|$ be the total number of z-values in $\mathcal{Z}$. 

First, for the root node, we count the frequency of each z-value in $\mathcal{Z}$ and store them in an array $C(\mathcal{Z})$ of size $M$.
The time of counting is $O(\mathcal{L}(\mathcal{Z}))$, and the space consumption of $C(\mathcal{Z})$ is $O(M)$.
Second, we find the most frequent z-value $z_1$. 
Let $\mathcal{Z}^{z_1}$ be the set of reference trajectories containing $z_1$. We build a node $e_1$ with label $z_1$ as a child node of the root, and put the trajectories in $\mathcal{Z}^{z_1}$ into the subtree of $e_1$.
Similarly, for $e_1$, we count the frequency of each z-value in $\mathcal{Z}^{z_1}$ and store them by an array $C(\mathcal{Z}^{z_1})$ with size $M$. 
Third, to find the next most frequent z-value $z_2$ in the remaining trajectories, we only need to compute $C(\mathcal{Z})-C(\mathcal{Z}^{z_1})$ rather than to count the frequency of each z-value in $\mathcal{Z}-\mathcal{Z}^{z_1}$. 
Thus, we find the most frequent z-value $z_2$ of $\mathcal{Z}-\mathcal{Z}^{z_1}$ according to $C(\mathcal{Z})-C(\mathcal{Z}^{z_1})$. Let $\mathcal{Z}^{z_2}$ be the set of reference trajectories containing $z_2$ in $\mathcal{Z}-\mathcal{Z}^{z_1}$. We build a node $e_2$ with label $z_2$ as a child node of the root and put the trajectories in $\mathcal{Z}^{z_2}$ into the subtree of $e_2$. 

We repeat this process when there is no trajectory left. Assume that we divide $\mathcal{Z}$ into $\mathcal{Z}^{z_1},\dots,\mathcal{Z}^{z_{t_1}}$, and we have $C(\mathcal{Z}^{z_1}),\dots,C(\mathcal{Z}^{z_{t_1}})$. That is, we have $t_1$ nodes in the first level of the RP-Trie. The time of finding all $z_i$ is $O(M{t_1})$.
The accumulated time of computing the arrays is bounded by $\sum_{i=1}^{t_1} O(\mathcal{L}(\mathcal{Z}^{z_i}))=O(\mathcal{L}(\mathcal{Z}))$. Therefore, the time cost of building the first level of the RP-Trie is $O(\mathcal{L}(\mathcal{Z})+M{t_1})$.

Similarly, if the $i$-th level has $t_i$ nodes, the time cost of building the $i$-th level is $O(\mathcal{L}(\mathcal{Z})+M{t_i})$. The optimized RP-Trie has at most $L+1$ layers and $\mathcal{L}(\mathcal{Z})+1$ nodes. Therefore, the overall time cost is $O(\mathcal{L}(\mathcal{Z})L+M\mathcal{L}(\mathcal{Z}))$. Normally, $L$ and $M$ are not large constants, $L\leq M$ and $ \mathcal{L}(\mathcal{Z})\leq NL$. So the time cost of building an optimized RP-Trie is $O(NM^2)$. In fact, most reference trajectories have lengths much less than $L$, and the number of the nodes is smaller than $\mathcal{L}(\mathcal{Z})$. So, the real time cost is much less than $O(NM^2)$.

\begin{table}[t]
	\centering
	\caption{Trajectories in $\mathcal{Z}$ to be Indexed}
	%\vspace*{0.1in}
	\label{tb:dataset_appendix}
	\resizebox{0.48\textwidth}{!}{
	\begin{tabular}{|m{0.4cm}<{\centering}|c|c||m{0.4cm}<{\centering}|c|c|}
		\hline
		\textbf{ID} & \multicolumn{2}{c||}{\textbf{Reference Trajectory}} & \textbf{ID} & \multicolumn{2}{c|}{\textbf{Reference Trajectory}}\\
		%	\hline
		\hline
		$Z_1$ &$ \{0001,0011\} $ &$\{1,3\}$ &$Z_5$ &$\{0011,0101\}$ &$\{3,5\}$    \\
		\hline
		$Z_2$ &$ \{0001,0011,0101\} $ &$\{1,3,5\}$ &$Z_6$ &$\{0001,0100\}$ &$\{1,4\}$    \\
		\hline
		$Z_3$ &$ \{0010,0011\} $ &$\{2,3\}$ &$Z_7$ &$\{0010,0100\}$ &$\{2,4\}$    \\
		\hline
		$Z_4$ &$ \{0010,0011,0101\} $ &$\{2,3,5\}$ &$Z_8$ &$\{0101,0110\}$ &$\{5,6\}$    \\
		\hline
	\end{tabular}
}
\end{table}

\begin{figure}
	\centering
	\includegraphics[width=.9\linewidth]{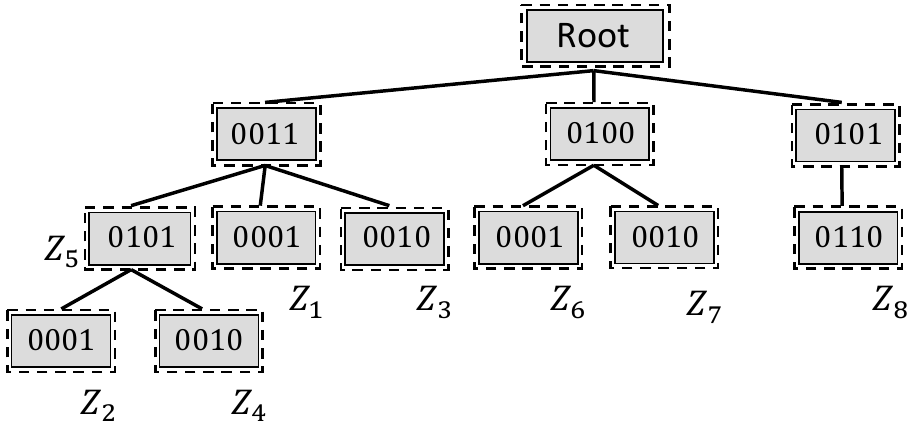}	
	\caption{Optimized RP-Trie for $\mathcal{Z}$}
	\label{fig:greedy_ex_trie}
\end{figure}

\begin{example}
	Consider a set of trajectories $\mathcal{Z}$ in Table \ref{tb:dataset_appendix}. There are $M=6$ cells, i.e., $\{0001,0010,0011,0100,0101,0110\}$. We have $C(\mathcal{Z})=\{3,3,5,2,4,1\}$. So the most frequent z-value is $0011$ with frequency 5. Let $z_1$ be $0011$ and $\mathcal{Z}^{z_1}=\{Z_1,Z_2,Z_3,Z_4,Z_5\}$. We build a node $e_1$ with label $0011$ and index these 5 trajectories in $e_1$. Then, we count the frequency of $\mathcal{Z}^{z_1}$ as $C(\mathcal{Z}^{z_1})=\{2,2,5,0,3,0\}$. To find $z_2$, we compute the frequency of $\mathcal{Z}-\mathcal{Z}^{z_1}$ as $C(\mathcal{Z})-C(\mathcal{Z}^{z_1})=\{1,1,0,2,1,1\}$. So the next most frequent z-value is $z_2=0100$ with frequency 2. 
	Similarly, $\mathcal{Z}^{z_2}=\{Z_6,Z_7\}$ and $C(\mathcal{Z}^{z_2})=\{1,1,0,2,0,0\}$. Now only one trajectory, $Z_8$, remains, and we arbitrarily select a z-value $0101$ in $Z_8$ as $z_3$. Thus, we finish building the first level of the RP-Trie.
	
	Then, we continue to build the subtrees of $e_1$, $e_2$ and $e_3$. As $C(\mathcal{Z}^{z_1})$, $C(\mathcal{Z}^{z_2})$ and $C(\mathcal{Z}^{z_3})$ have been obtained, the following computation is the same as that for $\mathcal{Z}$. Finally, we build an RP-Trie as shown in Fig. \ref{fig:greedy_ex_trie}.
\end{example}                                                           